\documentclass[american, aps, pra, reprint, superscriptaddress,  nofootinbib, longbibliography]{revtex4-1}
\usepackage[T1]{fontenc}
\usepackage[latin9]{inputenc}
\usepackage[usenames, dvipsnames]{color}
\usepackage{babel}
\usepackage{amsthm}
\usepackage{amsmath}
\usepackage{amssymb}
\usepackage{graphicx}
\usepackage{xcolor} 
\usepackage{array,tabulary} 
\usepackage[unicode=true,pdfusetitle, bookmarks=true,bookmarksnumbered=false,bookmarksopen=false,  breaklinks=false,pdfborder={0 0 0},backref=false,colorlinks=false] {hyperref}
\hypersetup{colorlinks,linkcolor=myurlcolor,citecolor=myurlcolor,urlcolor=myurlcolor}
\usepackage{soul,xcolor}

\makeatletter
\@ifundefined{textcolor}{}
{%
 \definecolor{BLACK}{gray}{0}
 \definecolor{WHITE}{gray}{1}
 \definecolor{RED}{rgb}{1,0,0}
 \definecolor{GREEN}{rgb}{0,1,0}
 \definecolor{BLUE}{rgb}{0,0,1}
 \definecolor{CYAN}{cmyk}{1,0,0,0}
 \definecolor{MAGENTA}{cmyk}{0,1,0,0}
 \definecolor{YELLOW}{cmyk}{0,0,1,0}
 }

\newtheoremstyle{nobold}
  {\topsep}
  {\topsep}
  {}
  {}
  {\itshape}
  {.}
  {.5em}
  {\thmname{#1}\thmnumber{ #2}\thmnote{ (#3)}}

\theoremstyle{nobold}

\theoremstyle{nobold}

\providecommand{\lemmaname}{Lemma}
\providecommand{\definitionname}{Definition}
\providecommand{\propositionname}{Proposition}

\newcommand{\nc}{\newcommand}
\nc{\rnc}{\renewcommand}

\usepackage{babel}
\usepackage{txfonts}
\usepackage{colortbl}\definecolor{myurlcolor}{rgb}{0,0,0.7}
\newcommand{\tr}{{\operatorname{Tr\,}}}

\def\ket#1{| #1 \rangle}
\def\bra#1{\langle  #1 |}
\def\braket#1{\langle  #1 \rangle}
\def\proj#1{| #1 \rangle\!\langle #1 |}
\nc{\ketbra}[2]{|#1\rangle\!\langle#2|}
\newcommand{\mean}[1]{\left\langle{#1}\right\rangle}

\DeclareMathOperator{\Tr}{\mathrm{Tr}}

\newcommand{\haH}

\newtheorem{theorem}{Theorem}

\newtheorem{observation}[theorem]{Observation}
\newtheorem{corollary}[theorem]{Corollary}

\newcommand{\norm}[1]{\left\| #1 \right\|}

\newcommand{\av}[1]{\left\langle #1 \right\rangle}
\newcommand{\dd}{\textrm{d}} 
\definecolor{orange}{RGB}{255,127,0}

\newcommand{\adop}{\hat a^{\dagger}}
\newcommand{\aop}{\hat a}
\newcommand{\be}{\begin{equation}}
\newcommand{\ee}{\end{equation}}
\newcommand{\dagg}{^{\dag}}
\newcommand{\im}{\text{i}}

\newcommand{\eqnref}[1]{Eq.~\eqref{#1}}
\newcommand{\eqnsref}[1]{Eqs.~\eqref{#1}}
\newcommand{\figref}[1]{Fig.~\ref{#1}}

\newcommand{\secref}[1]{Sec.~\ref{#1}}

\newcommand{\tabref}[1]{Table~\ref{#1}}

\begin{document}
\title{Bounds on the capacity and power of quantum batteries}
\author{Sergi Juli\`a-Farr\'e}
\email{sergi.julia@icfo.eu}
\affiliation{ICFO -- Institut de Ci\`encies Fot\`oniques, The Barcelona Institute of Science and Technology, ES-08860 Castelldefels, Spain}

\author{Tymoteusz Salamon}
\affiliation{ICFO -- Institut de Ci\`encies Fot\`oniques, The Barcelona Institute of Science and Technology, ES-08860 Castelldefels, Spain}
\affiliation{Institut f{\"u}r Teoretische Physik, Universit{\"a}t Ulm, Albert-Einstein-Allee 11, 89081 Ulm, Germany}

\author{Arnau Riera}
\affiliation{ICFO -- Institut de Ci\`encies Fot\`oniques, The Barcelona Institute of Science and Technology, ES-08860 Castelldefels, Spain}
\affiliation{Max-Planck-Institut f\"ur Quantenoptik, D-85748 Garching, Germany}

\author{Manabendra N. Bera}
\email{mnbera@gmail.com}
\affiliation{ICFO -- Institut de Ci\`encies Fot\`oniques, The Barcelona Institute of Science and Technology, ES-08860 Castelldefels, Spain}
\affiliation{Max-Planck-Institut f\"ur Quantenoptik, D-85748 Garching, Germany}
\affiliation{Department of Physical Sciences, Indian Institute of Science Education and Research (IISER), Mohali, Punjab 140306, India}

\author{Maciej Lewenstein}
\affiliation{ICFO -- Institut de Ci\`encies Fot\`oniques, The Barcelona Institute of Science and Technology, ES-08860 Castelldefels, Spain}
\affiliation{ICREA, Pg.~Lluis Companys 23, ES-08010 Barcelona, Spain}

\begin{abstract}
Quantum batteries, composed of quantum cells, are expected to outperform their classical analogs. The origin of such advantages lies in the role of quantum correlations, which may arise during the charging and discharging processes performed on the battery. In this theoretical work, we introduce a systematic characterization of the relevant quantities of quantum batteries, i.e., the \emph{capacity} and the \emph{power}, in relation to such correlations. For these quantities, we derive upper bounds for batteries that are a collection of non-interacting quantum cells with fixed Hamiltonians. The capacity, that is, a bound on the stored or extractable energy, is derived with the help of the energy-entropy diagram, and this bound is respected as long as the charging and discharging processes are entropy preserving. While studying power, we consider a geometric approach for the evolution of the battery state in the energy eigenspace of the battery Hamiltonian. Then, an upper bound for power is derived for arbitrary charging process, in terms of the  Fisher information and the energy variance of the battery. The former quantifies the speed of evolution, and the latter encodes the non-local character of the battery state. Indeed, due to the fact that the energy variance is bounded by the multipartite entanglement properties of batteries composed of qubits, we establish a fundamental bound on power imposed by quantum entanglement. We also discuss paradigmatic models for batteries that saturate the bounds both for the stored energy and the power. Several experimentally  realizable quantum batteries, based on integrable spin chains, the Lipkin-Meshkov-Glick and the Dicke models, are also studied in the light of these newly introduced bounds.
\end{abstract}

\maketitle
\section{Introduction}
With the decline of fossil fuels, there is a constant search for alternative energy sources. In this context, the growth of renewable energies has boosted the urgency for better energy storage devices, that is, batteries. They are often made-up of classical ingredients, be it chemical or cell-based, and are primary devices, where energy can be remotely stored and accessed deterministically. The reasons for storing energy are twofold. First, while energy disposal is aimed to be at will, renewable sources produce energy discontinuously, e.~g., solar panels do not generate energy at night. Second, in many cases, the power provided by such sources is not enough to perform some highly consuming tasks, like running a car. 

In recent years, there has been a growing interest in the study of the advantages that quantum effects could bring into the problem of energy storage \cite{Alicki13, 1367-2630-17-7-075015, Campaioli2017, PhysRevLett.120.117702, PhysRevA.97.022106,  PhysRevLett.122.047702, PhysRevB.98.205423, PhysRevB.99.035421, PhysRevLett.122.210601}, leading to the concept of \emph{quantum batteries}. They are intrinsically quantum devices made of quantum cells that can interact, and thus exploit collective quantum properties, in order to perform the task of energy storage. The study in Ref.~\cite{Alicki13}, for the first time, suggested that quantum entanglement can boost the extractable stored energy from an ensemble of quantum batteries. 
Later, it was shown that quantum entanglement is not absolutely necessary to increase the extractable energy, and classical correlations are enough \cite{PhysRevLett.111.240401}. Also, the presence of correlations, in the initial and final state of a quantum battery, is detrimental to its storage capacity. 

On the contrary, when looking at the power of a quantum battery, i.e., the rate at which energy can be stored or extracted, quantum correlations in the intermediate states can lead to an enhancement, usually denoted as a quantum speed up. In this line, the correspondence between quantum entanglement and the power of a quantum battery has led to many interesting studies, see for example \cite{Campaioli2017, PhysRevLett.120.117702, PhysRevA.97.022106, 1367-2630-17-7-075015, PhysRevLett.122.047702, PhysRevB.98.205423, PhysRevB.99.035421, PhysRevLett.122.210601}. The role of \emph{nonlocal} charging process has been studied theoretically in Refs.~\cite{Campaioli2017,1367-2630-17-7-075015}, and for experimentally realizable quantum batteries in \cite{PhysRevLett.120.117702, PhysRevA.97.022106}.  More precisely, one of the features that has been explored is the achievement of a superlinear rate of charging by means of collective quantum effects. In these cases, for a battery with $N$ quantum cells, the total power would scale as $N\sqrt{N}$, instead of (linear) $N$. One trivially gets the linear scaling in the case of independent charging of the quantum cells. A review on the recent progress can be found in Ref.~\cite{reviewbookQB}.

In this work, we deem to consider all these important aspects to characterize a quantum battery, adhering to the traditional definitions of the stored energy and the power. We study the capacity of a quantum battery, i.e., a bound on the storable (extractable) energy, under entropy-preserving processes. Notice that the scenario of an entropy-constrained capacity was introduced in Ref.~\cite{Allahverdyan_2004} for the study of work extraction in finite quantum systems. Here, we extend these results for the study of quantum batteries. In particular, we introduce the energy-entropy diagram for quantum states, and we show its usefulness to visualize both the limitations on extractable and storable work. We also emphasize that many of the previous works about quantum batteries assume only the expectation value of the battery Hamiltonian, i.e., the first moment. However, we know that this expectation value does not always imply accessible energy, as a non-vanishing variance implies lower usability of that energy \cite{Friis2018precisionwork}.  Therefore it is important to study higher moments of the energy, in order to properly characterize a quantum battery.  

In the case of power, one of the main results is a bound on the rate at which energy can be deposited (extracted) in a quantum battery in a charging (discharging) process, obtained by means of a quantum geometrical approach. It is valid for arbitrary charging (discharging) processes. The bound is derived in terms of the energy variance of the battery and the Fisher information (or speed of evolution) in the eigenspace of the battery Hamiltonian (denoted as $I_E$). On the one hand, the energy variance of the battery can be related to non-local properties of its quantum state, hence connecting entanglement and power. In the case of qubit-based batteries, this connection leads to another central result of this work, that is, a mathematical inequality bounding the power by the amount of entanglement generated between the cells. On the other hand, $I_E$ signifies the rate of change of the battery state in the energy eigenspace. This gives a better bound on power compared to the case in which one considers the traditional speed of evolution of the battery state in Hilbert space (i.e., the quantum Fisher information, $I_Q$). The reason is related to the fact that there can exist initial and final time-evolved states that are infinitesimal close (or identical) in energy, but orthogonal (and thus perfectly distinguishable) in Hilbert space. The speed of evolution in Hilbert space would then be non-zero, while power would be zero.  Therefore, when studying power, it is necessary to consider a speed of evolution based on a notion of distinguishability between states that is directly connected to the difference in their energetic distributions.

Furthermore, we use the derived bounds for the storable energy and the power to systematically analyze the paradigmatic cases that are often studied the literature, and also more realistic models of quantum batteries, that include integrable spins chains, and two models based on the LMG and the Dicke Hamiltonian. 
For such a study, we restrict to charging processes in which the battery is initially in a pure quantum state. This leads to an entropy-free capacity given simply by the difference between the ground state and the highest-excited energies of the battery Hamiltonian. We then analyze which portion of this capacity is stored in the battery due to the charging dynamics, and its scaling with the number of cells $N$. In the case of power, we also analyze the saturation of the bound in terms of the quantities appearing in it and their scaling with $N$.

The main results for this part are the saturation of our bounds in the paradigmatic cases, the unveil of the role of entanglement in the power of all these models, and the challenging of the appearance of a quantum speed up in the realistic models that we consider.

The article is organized as follows. In Sec.~\ref{sec:QuBatProp}, we outline the important properties of a quantum battery. Section \ref{sec:BoundCapacity} gives the expression for capacity exploiting the energy-entropy diagram.   
In Sec.~\ref{sec:BoundPower}, a bound on the power of a generic battery is derived, in terms of the energy variance and the Fisher information in energy eigenspace. In the case of qubit-based quantum batteries, we introduce a bound on power in terms of the $k$-qubit entanglement of the battery state. Section \ref{sec:ParadigmaticExamples} outlines paradigmatic spin-$\frac{1}{2}$(qubit) models that saturate the derived bounds. We study various realistic models in light of these bounds in Sec.~\ref{sec:BoundsSpinModel}. Finally, we dedicate Sec.~\ref{sec:Discussion} for conclusions.

\section{Quantum Batteries and their properties \label{sec:QuBatProp}}

In this section, we make a brief outline of the properties that one needs to consider to make an assessment of a good battery. A quantum battery is a physical system, where energy can be stored for a relatively long time and extracted whenever it is convenient. It is modeled by a Hamiltonian $H_B$, so that its internal energy, which depends on its state $\rho$, is given by
\begin{equation}
E(\rho)\coloneqq \Tr (\rho H_B)\, .
\end{equation}
In order to give a further insight into the problem, in this work, we assume that a battery is composed of independent noninteracting quantum cells, with the Hamiltonian
\begin{equation}\label{eq:battery_ham}
H_B= \sum_{j=0}^{N-1} h_j,
\end{equation}
where $ h_j$ is the Hamiltonian of the $j$-th quantum cell, and $N$ is the total number of them. Notice that the form of the battery Hamiltonian in \eqnref{eq:battery_ham} implies that we are considering an additive nature of the battery stored energy, i.e., the total energy is obtained as the sum of individual energies stored in each battery cell.

The process of charging (or discharging) a battery is a physical process $\Gamma_t$ that leads to the battery state $\rho(t)\coloneqq \Gamma_t(\rho(t_0))$ at each instant of time $t$. The main features that we study in this work are the following.

\

\noindent {\bf Stored and extracted energy} -- For a given dynamical charging (discharging) process $\Gamma_t$ of a quantum battery, that is initially in the state $\rho(t_0)\coloneqq\rho_0$, the stored (extracted) energy $E^s_{\rho_0}\ (E^e_{\rho_0})$ is the maximum amount of energy that the battery absorbs (delivers). They are defined by 
\begin{equation}\label{def:capacity}
\begin{split}
E^s_{\rho_0}[\Gamma_t]&\coloneqq \max_t E\left(\rho(t)\right)-E\left(\rho_0\right),\\
E^e_{\rho_0}[\Gamma_t]&\coloneqq E\left(\rho_0\right)-\min_t E\left(\rho(t)\right).\\
\end{split}
\end{equation}
Notice that, in the above definitions, the energies are a functional of the process, and the initial battery state (i.e., we do not assume that the battery is initially fully empty nor charged). Instead, we are interested in the energy change of the battery given an arbitrary initial state and a dynamical process $\Gamma_t$, belonging to a certain set $\mathcal{S}$ of control operations. Within this set, we also define the capacity $C_\mathcal{S}$ of the battery as its energetic amplitude, that is 
\begin{equation}
C_{\mathcal{S}} \coloneqq \max_{\Gamma_t \in \mathcal{S}} E^s_{\rho_0}[\Gamma_t] + \max_{\ \Gamma_t \in \mathcal{S}} E^e_{\rho_0}[\Gamma_t].
\end{equation}
In the present work, we will mainly focus in entropy-preserving operations, which can be connected to unitary operations for linear maps or in the many-copy limit, as discussed in Sec.~\ref{sec:BoundCapacity}. Notice then, that the capacity $C_\mathcal{S}$ will depend on the initial state $\rho_0$ trough the initial entropy $S(\rho_0)$. Physically, this means that, while given a battery Hamiltonian $H_B$ the ideal capacity is trivially the difference between the ground state and maximally excited state, one may not be able to bring the battery in an exact pure state, thus reducing the effective capacity under entropy-preserving/unitary operations \cite{Allahverdyan_2004}.

The choice of this set of control operations is motivated by the fact that batteries are not engines that convert heat into work. Instead, they are supposed to operate in isolation from the environment and store or supply energy when demand arises. Also notice that within entropy-preserving operations the stored energy, as defined in Eq.~\eqref{def:capacity}, coincides with the (thermodynamical) work \cite{Bera2019thermodynamicsas}. If we go beyond this assumption, we have to consider work injection instead of internal energy, as they migh not coincide.  In this context, notice that recent works directly addressed the issue of non-unitary processes by considering dissipative charging protocols \cite{PhysRevB.99.035421,PhysRevLett.122.210601} or the impact of correlations in the extractable work from the battery \cite{PhysRevLett.122.047702}.
It is also important the fact that even when restricting to entropy-preserving or unitary operations, we only consider the first moment of the Hamiltonian. However, a nonzero variance in the stored energy worsens its deterministic extraction, as in a discharging process it might not be possible to exactly reverse the time evolution generated during the charging.

In this work, we will not take this variance into account to derive our bound for storage. The reason is that for a large number of cells $N$ the energy variance (second moment) typically scales as $1/\sqrt{N}$, and thus vanishes in the thermodynamic limit. On the contrary, we will analyze the impact of this variance when studying the charging processes of different particular models of a quantum battery. An enhanced energy variance in the battery may appear during the time evolution. In these cases, it is important to see what percentage of it remains in the final state, and how does it decay in the large $N$ limit. \\

\noindent {\bf Power} -- How quickly a battery can be charged (or discharged) depends on its power. It is quantified by the rate of energy flow in the battery during charging (or discharging), that is
\begin{equation}
P(t)\coloneqq \frac{\dd}{\dd t} E\left(\rho(t)\right)\, .
\end{equation}
Notice that as in the capacity definiton, the above definition of Power coincides with the rate of Work injection only in the case of entropy-preserving transformations of the battery state, that is when we considered an input of energy in an isolated system, without heat exchange. This is also equivalent to an adiabatic process where the system undergoes a transformation without an interaction with the environment. For a given charging process of a duration $\Delta t = t_f-t_0$, the \emph{average power} will be
\begin{equation}
\langle P \rangle_{\Delta t}= \frac{E(\rho(t_f))-E(\rho(t_0))}{\Delta t}\, .
\end{equation}
From a geometric point of view, there are two relevant properties when studying how fast a charging process can evolve a battery state from $\rho(t_0)$ to $\rho(t_f)$. Within a notion of distance (distinguishability)  between quantum states, the amount of time spent in the process is affected by  both the rate at which the battery state $\rho$ changes during the time interval $\Delta t$ (speed of state evolution), and how smart is the path taken, in terms of the total length (trajectory of evolution). These two aspects are illustrated in Fig.~\ref{fig:GeometryP}.
\begin{figure}[h!]
\includegraphics[width=0.8\columnwidth]{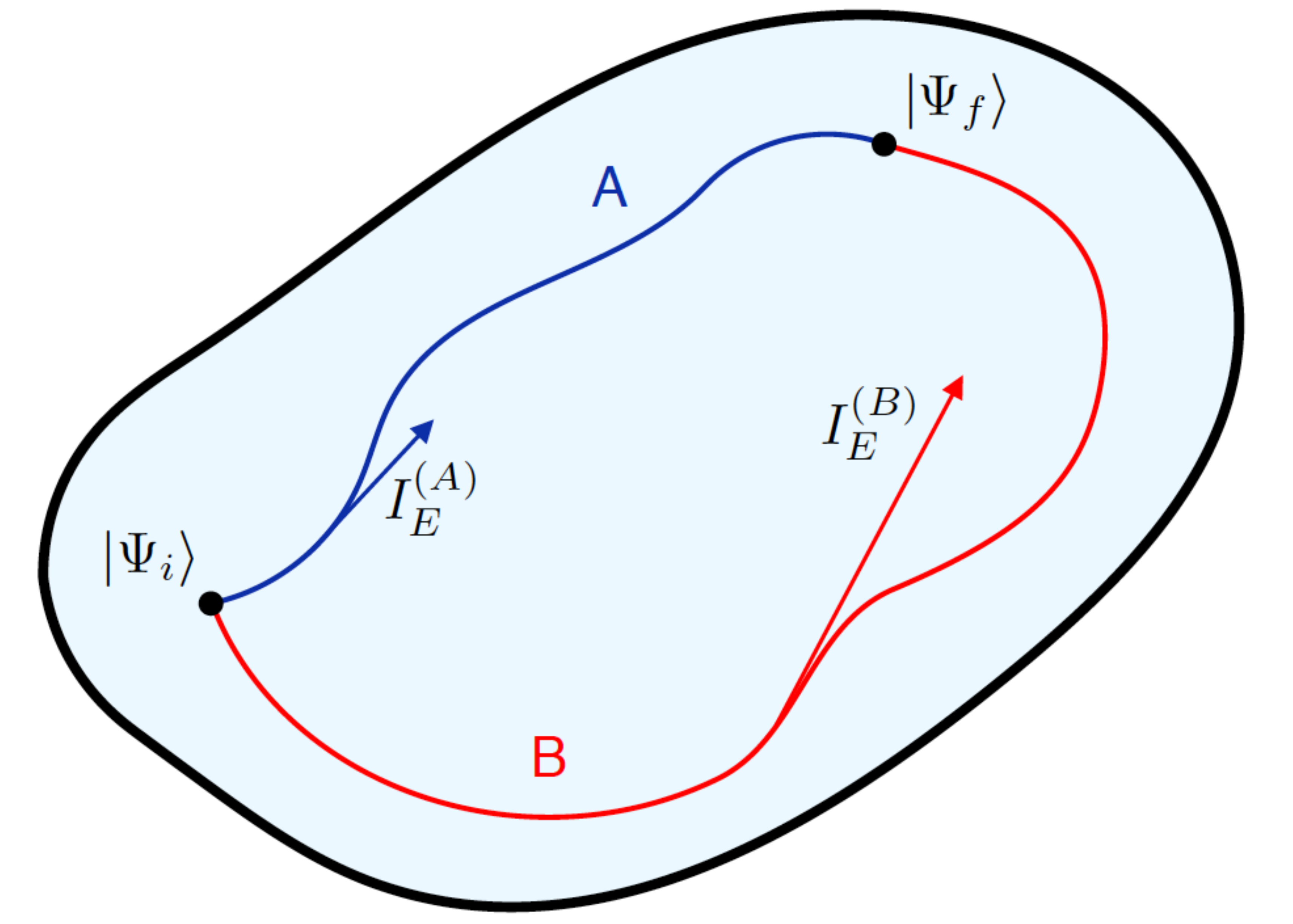}
\caption{Schematic of a charging (discharging) process in the quantum state space. The trajectories represent two different charging processes $A$ (blue) and $B$ (red) in which the initial and final states ($\ket{\Psi_i}$ and $\ket{\Psi_f}$ respectively) coincide. The process $A$ undergoes a path with path length $L_A$ at a speed $I_E^{(A)}$. Similarly, the process $B$ traverses the path length $L_B$ and the speed $I_E^{(B)}$. The time required ($t_A$ and $t_B$ for the two processes) to reach to the final state, from the initial one, depends both on the path length and the speed. For instance, the charging process requires shorter time, i.e. $t_A < t_B$, although the speed of the process $B$ is larger,  $I_E^{(A)}<I_E^{(B)}$. This is because path length of $B$ is also larger, $L_A < L_B$.}
\label{fig:GeometryP}
\end{figure}
Notice that the concepts of speed and trajectory of evolution are highly dependent on the definition of distinguishability between battery states. For instance, if one considers distinguishability of states in Hilbert space, an initial state could evolve to a perfectly distinguishable final state, but with the same energy distribution as the initial one. In this case, the speed of evolution would be nonzero, without a change in the energetic properties of the state. Instead, in this work, we use the concept of distinguishability in the eigenspace of the battery Hamiltonian and define the speed of evolution accordingly.

In the following, we aim to understand the limits that quantum mechanics imposes on the storage, the power, and energy variance of any energy storage device. More specifically, we are interested in understanding how these quantities scale with the number of cells $N$. In doing so, we restrict ourselves to entropy-preserving evolutions to perform charging and discharging processes on the batteries.

\section{Bounds on the stored and extracted energies \label{sec:BoundCapacity}}
Here we study the capacity of a battery, given by the maximum amount of energy that it can store or that can be extracted from it. We restrict our analysis to entropy-preserving processes, which reduce to unitary operations in the case of linear maps acting on a single copy of a state  \cite{Hulpke2006}. In the many-copy scenario, entropy-preserving operations are understood as effective local operations generated by a global unitary acting on asymptotically many copies of the system \cite{Sparaciari16, Bera2019thermodynamicsas} and a small ancilla $\eta$. That is, for any two states $\rho$ and $\sigma$ with equal entropies, $S(\rho)=S(\sigma)$, (here we consider the von Neumann entropy $S(\rho)=-\tr \rho \log_2 \rho$) there exists an additional ancilla system of $O(\sqrt{N\log_2 N})$ qubits and a global unitary $U$ such that
\begin{equation}\label{eq:entropy-pres-micro}
\lim_{N\to \infty}
\|\tr_{\textrm{anc}}\left(U \rho^{\otimes N}\otimes\eta U^{\dagger}\right)-\sigma^{\otimes N}\|_1=0  \, ,
\end{equation}
where the partial trace is performed over the ancillary qubits and $\|\cdot\|_1$ is one-norm. In the finite case, where one cannot apply the limit appearing in Eq.~\eqref{eq:entropy-pres-micro}, we will also discuss how the operational interpretation of our analysis changes.

\subsection{Energy-entropy diagram}
In order to understand the limitations on the energy that a system can store, it is useful to introduce the \emph{energy-entropy diagram}, as depicted in Fig. \ref{fig:energy-entropy-diagram}. Given a system described by a time-independent Hamiltonian $H$, a state $\rho$ is represented in the energy-entropy diagram by a point with coordinates $x_\rho\coloneqq(E(\rho), S(\rho))$ (see Fig.~\ref{fig:energy-entropy-diagram}). All physical states reside in a region that is lower bounded by the horizontal axis (i.e., $S=0$) corresponding to the pure states, and upper bounded by the convex curve $(E(\beta),\ S(\beta))$ which represents the thermal states of both positive and negative temperatures. Let us denote such a curve as the thermal boundary. The inverse temperature associated with one point of the thermal boundary is given by the slope of the tangent line in such a point, since 
\begin{equation}
\frac{\dd S(\beta)}{\dd E(\beta)}=\beta\, .
\end{equation}
Notice that here a thermal state with negative $\beta$ corresponds
to the case of population inversion, that is, the probability of
finding the system in a given energy eigenstate increases with
increasing energy.

Note also that a point $x_\rho\coloneqq(E(\rho),\ S(\rho))$ of the energy-entropy diagram corresponds in general to several quantum states since, given a Hamiltonian, there are several states with equal energies and entropies.

\begin{figure}
\includegraphics[width=\columnwidth]{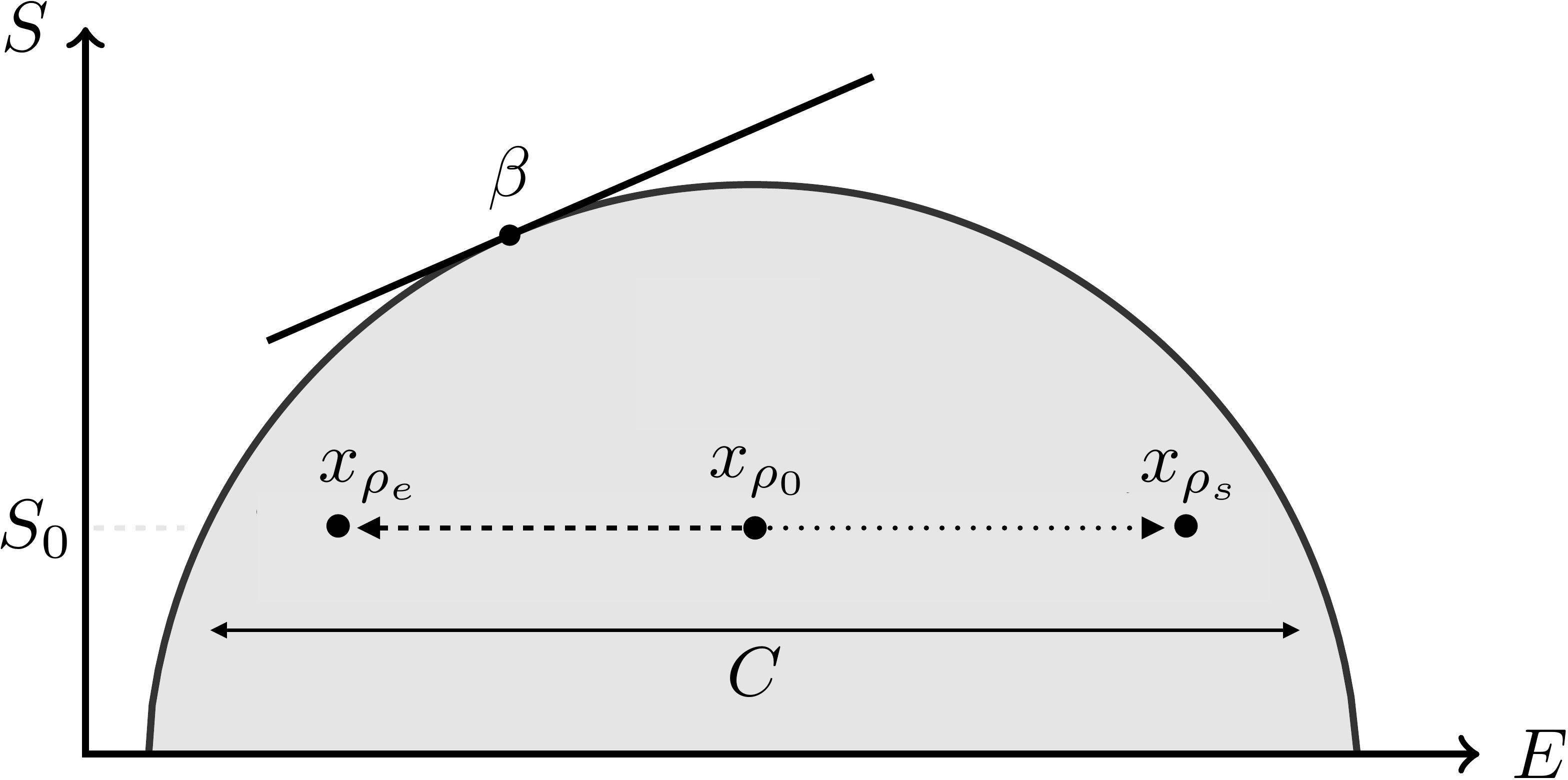}
\caption{
Energy-entropy diagram. Any quantum state $\rho$ is represented in the diagram
as a point with coordinates $x_\rho\coloneqq(E(\rho), S(\rho))$. 
The point in the thermal boundary with slope $\beta$ corresponds to the thermal state with inverse temperature $\beta$.
The trajectory of a charging (discharging) process is represented with a dotted (dashed) line. They connect, respectively, the initial state $\rho_0$, with entropy $S_0$, and the final state $\rho_s$ ($\rho_e$).  The capacity $C$ corresponds to the energetic amplitude at a given entropy. }
\label{fig:energy-entropy-diagram}
\end{figure}
 \subsection{Capacity under entropy-preserving operations}
Now, with the energy-entropy diagram, we derive the capacity $C(S)$ of a quantum battery. A dynamical process that changes either the energy or the entropy of a system is represented as a trajectory in the energy-entropy diagram. 
As we are restricted to charging (discharging) processes that are 
entropy-preserving, the trajectories of such a processes will be horizontal lines in the diagram, see Fig.~\ref{fig:energy-entropy-diagram}.

 \begin{observation}[Capacity]\label{thm:capacity}
Consider an initial battery state $\rho_0$ and a dynamical process $\Gamma_t$, within the set of entropy-preserving maps. The stored and extracted energies are upper bounded by
\begin{equation}\label{eq:capacity}
\begin{split}
E^s_{\rho_0}[\Gamma_t] & \leqslant E_{\max}(S(\rho_0))-E(\rho_0),\\
E^e_{\rho_0}[\Gamma_t] & \leqslant E(\rho_0)-E_{\min}(S(\rho_0)).
\end{split}
\end{equation}
If we only fix the initial entropy $S_0$, there is a unique bound given by addition of the two above:
\begin{equation}
E^s_{\rho_0}[\Gamma_t]+ E^e_{\rho_0}[\Gamma_t] \leqslant C(S_0)
\end{equation}
where $C(S)$ is the entropy-dependent capacity of the battery
\begin{equation}\label{eq:capacity2}
C(S)  =  E_{\max}(S)-E_{\min}(S),
\end{equation}
with
\begin{equation}\label{eq:optimization}
E_{\min/\max}(S)\coloneqq \min_{\sigma \colon S(\sigma)=S}/\max_{\sigma \colon S(\sigma)=S} E(\sigma)\, .
\end{equation}
Here the minimization/maximization is made over all states with entropy $S$. Note that the minimum/maximum in \eqref{eq:optimization} is achieved by thermal states with positive/negative $\beta$ also called completely passive/completely active states \cite{Pusz1978}.
\end{observation}
\begin{proof}\renewcommand{\qedsymbol}{}
 The proof is straightforward from the energy-entropy diagram and the fact that the processes are entropy-preserving.   
\end{proof}	
A first insight from the energy-entropy diagram is that intersystem correlations are not needed in order to saturate the optimal capacity of a battery. The two states that respectively minimize and maximize the energy given an entropy lie in the thermal boundary and are therefore thermal states.  As thermal states of non-interacting Hamiltonians between different cells are product states, classical or quantum correlations do not provide an advantage in saturating the capacity bound.

Let us finally discuss the operational interpretation of the bounds (\ref{eq:capacity}). In the many-copy limit, for a fixed Hamiltonian and a given initial entropy, the completely passive and completely active states, corresponding to $E_{\min}$ and $E_{\max}$, are always reachable with the help of a small ancilla and a global unitary operation. However, this statement is not always true in the few-copy case, where unitary operations represent only a subset of entropy-preserving ones. In this latter case, even though the bounds (\ref{eq:capacity}) are still respected, it will not always be possible to saturate them when the set of control operations is limited to unitary operations \cite{Niedenzu2019conceptsofworkin}.

\section{Bound on power \label{sec:BoundPower}}
In this section, we derive the bound on power, following geometric approaches towards quantum speed and trajectories. Although our discussion is focused on the problem of energy storage, our bound can also be applied to any other observable $O$, when the main goal is to increase its expectation value $\av{O(t)}$ as fast as possible. 
\subsection{Speed of evolution in state space and energy eigenspace \label{sec:GeoSpeed}}
We start by introducing a notion of distance between states in Hilbert space. At this point we deem appropriate to mention that many of the concepts that will appear in this section appear naturally in the framework of quantum speed limits (QSL) (for a review, see Ref.~\cite{Deffner_2017}). However, here we will use them in the form crafted for the study of energy storage. Let us consider the \emph{Bures angular distance} \cite{bengtsson_zyczkowski_2017}, between two quantum states $\rho$ and $\sigma$, defined as 
\begin{equation}
D_{Q}(\rho, \sigma)=\arccos \left[F(\rho,\sigma)\right]\, ,
\end{equation}
where $F(\rho,\sigma)= \tr \left(\sqrt{\sqrt{\rho}\sigma\sqrt{\rho}}\right)$ is the Uhlmann's fidelity \cite{Nielsen00}. Now, for an evolution of a system $\rho(t) \rightarrow \rho(t+dt)$, the instantaneous speed in state space is defined as
\begin{equation}
v(t)\coloneqq \lim_{\delta t\to 0} \frac{D(\rho(t+\delta t),\rho(t))}{\delta t}\, .
\end{equation}
After a straightforward calculation (see, for example, Ref.~\cite{Deffner_2017}), it can be rewritten as
\begin{align}
v(t)=\frac{1}{2}\sqrt{{I}_Q(\rho(t))},
\end{align}
where $\mathcal{I}_Q(\rho(t))$ is the \emph{quantum Fisher information} (QFI).  
For any quantum state $\rho(t)=\sum_i p_i \proj{i}$, which is undergoing a unitary evolution driven by a Hamiltonian $H(t)$, the QFI is given by
\begin{align}
 {I}_Q(\rho(t))=\sum_{i\neq j}\frac{(p_i-p_j)^2}{p_i+p_j} |\braket{i|H(t)|j}|^2.
\end{align}
The QFI, in information theory, has the interpretation of an information measure \cite{PhysRevLett.72.3439}. In fact, the ${I}_Q(\rho(t)) dt^2$ quantifies the distance between states $\rho(t)$ and $\rho(t+dt)$ that are separated by an infinitesimal time $dt$ and driven by the Hamiltonian $H(t)$.  In the context of quantum metrology, a higher value of QFI indicates a potential to result in a higher precision estimation of a parameter. For instance, $\mathcal{I}_Q(\rho(\theta))=0$ implies that any information about the parameter $\theta$ cannot be extracted, whereas divergent $\mathcal{I}_Q(\rho(\theta))\to \infty$
means estimation of the parameter $\theta$ with infinite precision. We refer to the Refs. \cite{Toth2014} and \cite{Kolodynski2015} for a review in the context of quantum Information and quantum Optics, respectively.

For the case of pure states, the Bures distance reduces to the Fubini-Study distance, which is given by
\begin{equation}\label{eq:Fubini-Study}
D(\psi,\phi)\coloneqq \arccos|\bra{\phi}\psi\rangle|\, .
\end{equation}
Then, the corresponding speed, for the case of a unitary time evolution  driven by a Hamiltonian $H_C(t)$, becomes
\begin{equation}
v(t)=\sqrt{\bra{\psi(t)}H_C(t)^2\ket{\psi(t)}-\bra{\psi(t)}H_C(t)\ket{\psi(t)}^2}\eqqcolon \Delta H_C(t) \, .
\end{equation}

Hence, for pure states, the speed of the system in the Hilbert space is given by the instantaneous energy variance measured by the charging Hamiltonian $H_C(t)$ that drives the evolution of the system. Note that, in the case of a time-independent charging Hamiltonian, the energy variance does not change during the entire evolution.
Once we have a notion of speed, the length of the trajectory followed for a time $t_F$ is given by
\begin{equation}
\label{eq:LengthTrajectory}
L[\rho(t), t_F]= \int^{t_F}_0 \dd t \ v(t)\, .
\end{equation}

\begin{figure}
\includegraphics[width=0.85\linewidth]{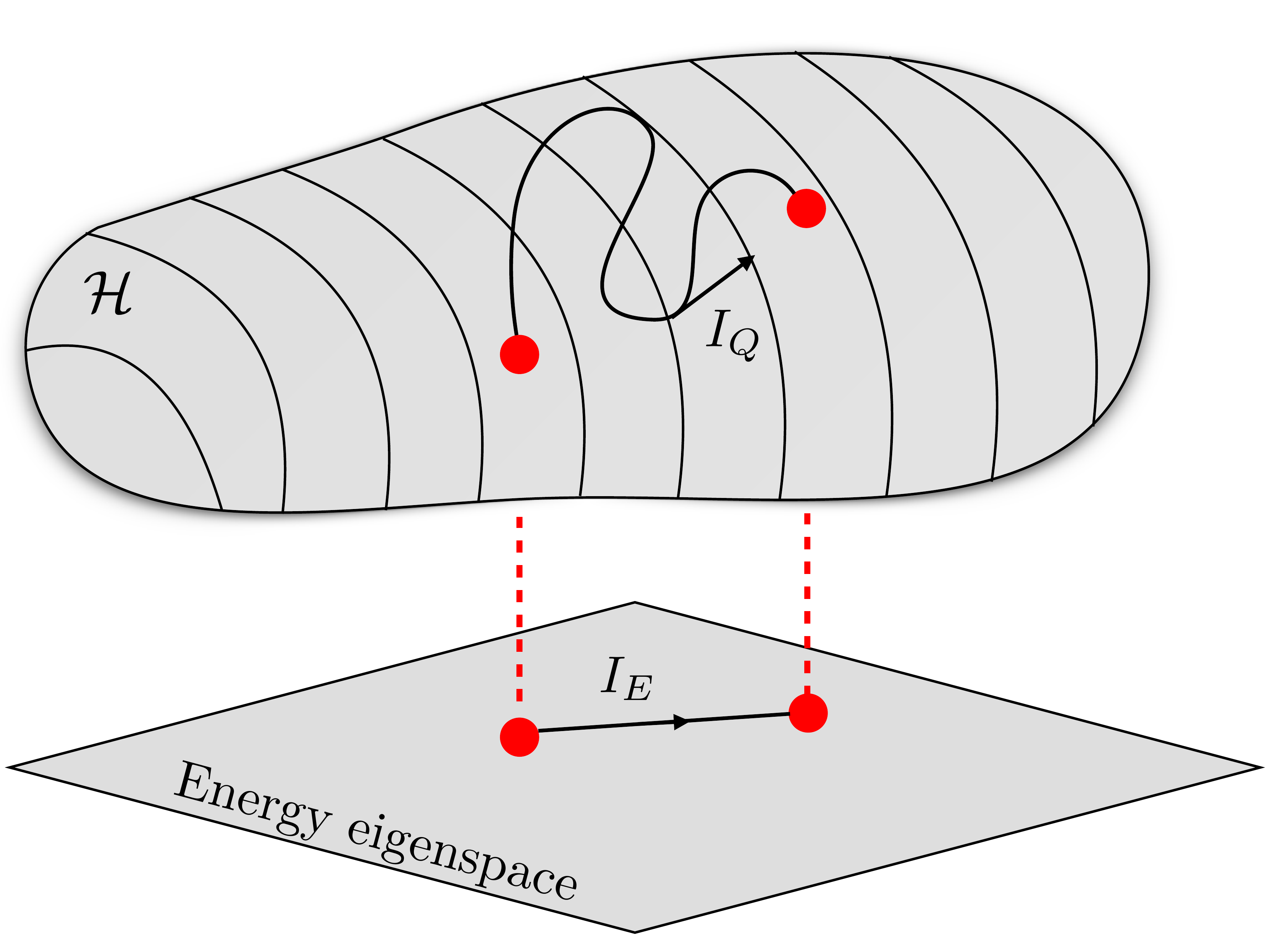}
\caption{A schematic representation of speed of quantum evolution in Hilbert space ($I_Q$) and the speed of evolution in the energy eigenspace of the battery Hamiltonian ($I_E$). Notice that, in general, $I_Q \geqslant  I_E $.} 
\label{fig:ProjectedI}
\end{figure}

When considering the Bures angle as a measure of distance, we are looking at the physical distinguishability of the system. However, when looking at systems as quantum batteries, i.e., energy storage devices, we are not interested in how fast the state changes, but in how fast its energy distribution evolves. In other words, there are orthogonal states (perfectly distinguishable) that have identical energy distributions. Thereby, although the system can be moving very fast in the state space, its change in the energy content can be negligible. From this perspective, it is useful to introduce a measure of distance between quantum states not based on their statistical distinguishability, but in their energetic distinguishability (see Fig.~\ref{fig:ProjectedI}).
To do so, let us write the battery Hamiltonian in its spectral representation 
\begin{equation}
H_B=\sum_k E_k P_k\, ,
\end{equation}
where $P_k$ is the projector onto the eigenspace associated to the eigenvalue $E_k$. The energy distribution of a state $\rho$ is given by the populations
\begin{equation}\label{eq:def_pk}
p_k\coloneqq \Tr\left(P_k\, \rho\right)\, .
\end{equation}
The speed in the energy space can then be defined as the relative entropy distance between the energy distributions in two consecutive moments of time
\begin{equation}
v_E(t)\coloneqq \lim_{\delta t\to 0} \frac{D_{KL}(\vec p(t+\delta t)\|\vec p(t))^{1/2}}{\delta t}\, ,
\end{equation}
where $D_{KL}(\vec p,\vec q)$ is the relative entropy or Kullback-Leibler divergence \cite{bengtsson_zyczkowski_2017} between two discrete probability distributions $\vec p$ and $\vec q$. It is defined by 
\begin{equation}
D_{KL}(\vec p\|\vec q)\coloneqq \sum_k p_k \log_2\frac{p_k}{q_k}\, .
\end{equation}
After a straightforward calculation, one gets
\begin{equation}
D_{KL}\left(\vec p(t+\delta t)\|\vec p(t) \right)= \sum_k\frac{\dot p_k^2}{p_k}\delta t^2 +O(\delta t^3)\, ,
\end{equation}
where the first order contribution vanishes due to $\sum_k \dot p_k =0$. This means that the right distance to define a speed is the square root of the relative entropy, which can be written in terms of the \emph{Fisher information in the energy eigenspace} as
\begin{equation}
I_E(t)  :=  2v_E(t)^2=\sum_k \left( \frac{\dd}{\dd t}\log_2 p_k(t)\right)^2 p_k(t)\, .
\end{equation}
It is interesting to point out that the same conclusion is reached when, instead of using the Kullback-Leibler distance, one employs the angular distance $D(p,q)\coloneqq \arccos F_{\textrm{cl}}(p,q)$, with $F_{\textrm{cl}}(p,q)=\sum_k \sqrt{p_k q_k}$ being the classical fidelity. As the quantum fidelity reduces to the classical one in the case where the states are diagonal in the same eigenbasis, the speed in energy space can be understood as the speed in the state space of the dephased states in the energy basis
\begin{equation}
v_E(t)=  \lim_{\delta t\to 0} \frac{D\left(\overline{\rho(t+\delta t)},\overline{\rho(t)}\right)}{\delta t} \leqslant v(t)\, ,
\end{equation}
where $\bar \rho \coloneqq \sum_k P_k\rho P_k$ represents the dephased state in the energy eigenbasis. We refer to the Fig.~\ref{fig:ProjectedI} for a schematic representation. The last inequality is a consequence of the Bures distance being monotonically decreasing under quantum operations.

Note that both the $I_{Q}$ and the $I_E$ of uncorrelated and independent systems are additive. Thus, for a system composed on $N$ identical subsystems in which each subsystem goes through the same independent evolution, the speed at which the system runs along a trajectory scales as $\sqrt{N}$. This scaling for independent subsystems will be relevant in the later discussion.

\subsection{The bound on power}
Equipped with this geometric framework, let us introduce the following result. That is an upper bound on the rate at which any dynamical process can change the mean value of a given moment of an observable.

\begin{theorem} \label{thm:main}
Given an observable $O$, that is time-independent  in the Schr\"odinger picture, the following inequality is satisfied
\begin{equation}
\left(\frac{\text{d}}{\text{dt}}\langle{O^m}\rangle\right)^2 \leq \Delta {\left(O^m\right)^2}I_O(t), 
\end{equation}
where $\Delta {\left(O^m\right)^2}$ is the variance of the $m$-th moment of the observable that captures how non-local the evolution process is, and $I_O(t)$ is the Fisher information, which corresponds to the speed of the process in the observable ($O$) eigenspace.
\end{theorem}
\begin{proof}\renewcommand{\qedsymbol}{}
We first write the $m$-th power of the operator $O$ in its spectral decomposition:
\be 
O^m=\sum_k O_k^m \Pi_k,
\ee
where $O_k$ are the eigenvalues of $O$ and $\Pi_k$ the projectors onto the corresponding subspaces. Using this decomposition, we can write the expected value of $O^m$ as 
\be 
\langle O^m\rangle = \sum_k O_k^m \pi_k(t),
\ee
where $\pi_k(t) \coloneqq \tr{(\rho(t)\Pi_k)}$. Taking the time-derivative of the last equation we get
\be 
\frac{\text{d}}{\text{dt}}\langle O^m\rangle = \sum_k O_k^m \dot{\pi}_k(t)=\sum_k \sqrt{\pi_k(t)}\left(O_k^m-C(t)\right)\frac{\dot{\pi}_k(t)}{\sqrt{\pi_k(t)}},
\ee
where in the last step we have used $\sum_k\dot{\pi}_k(t)=0$. Finally, using the Cauchy-Schwarz inequality, we get
\be 
\left(\frac{\text{d}}{\text{dt}}\langle O^m\rangle\right)^2 \leq \left[\sum_k \pi_k(t)\left(O_k^m-C(t)\right)^2\right]\left[\sum_l\frac{\dot{\pi}_l(t)^2}{{\pi_l(t)}}\right].
\label{eq:cauchy-schwarz_general}
\ee
We can now identify the second factor on the right-hand side as the \emph{Fisher information} $I_O$, representing the speed of evolution in the observable $(O)$ eigenspace. 
Furthermore, as Eq.~\eqref{eq:cauchy-schwarz_general} is valid for any $C(t)$, a minimization over $C$ leads to
\begin{equation}
C= \langle O^m\rangle,
\end{equation}
which leads us to identify the first factor on the right-hand side as $\Delta {\left(O^m\right)^2}$, and completes the proof.
\end{proof}
 
\begin{corollary}[Power]  \label{cor:PowerBound}
Given a process for charging (or discharging) a battery, with Hamiltonian $H_B$, its instantaneous power fulfills 
\begin{equation}\label{eq:main_thm}
P(t)^2\leqslant \Delta H_B(t)^2 I_E(t),
\end{equation}
where $\Delta H_B(t)^2$ is the variance of the battery Hamiltonian that captures how non-local in energy the charging process is, and $I_E(t)$ is the Fisher information, which corresponds to the speed of the charging process in the energy eigenspace.
\end{corollary}

The corollary above can be seen as a special case of the theorem \ref{thm:main}, where the observable is the battery Hamiltonian, $O=H_B$, and $m=1$. We can parametrize the tightness of the bound for power via an angle $\theta_P$ that satisfies
\be 
\cos{\theta_P}\coloneqq\frac{P}{\sqrt{\Delta H_B^2 I_E^2}}.
\ee
This angle may be used to quantify how efficient a charging process is in terms of power, if one considers $\Delta H_B^2$ and $I_E$ as resources that can give maximum power when $\cos(\theta_P)=1$. Furthermore, in some cases, it could be useful to consider a time-averaged version of the bound (\ref{eq:main_thm}) to eliminate the time-dependence. One possibility would be to consider the bound
\be 
\frac{\Delta E}{\Delta t} \leq \sqrt{\mean{\Delta H_B^2}_{\Delta t} \mean{{I_E}}_{\Delta t}}\, , 
\ee
where $\Delta E$ is the change in the battery energy during the interval $\Delta t$, and 
\begin{equation}\label{eq:TimeAv}
\mean{X}_{\Delta t} \coloneqq 1/\Delta t \int_{t_0}^{t_0+\Delta t}\  X\  dt. 
\end{equation}

To show that the speed of evolution in energy eigenspace $I_E$ is more informative than the one in Hilbert space when studying power, let us consider that the charging process is driven by a Hamiltonian evolution given by $H_C(t)$. We can derive the following inequality from Heisenberg's uncertainty principle:
\begin{equation} 
\left(\frac{\text{d}}{\text{dt}}\mean{H_B}\right)^2=|\mean{[H_B,H_C(t)]}|^2\leq 4\Delta H_B^2\Delta H_C(t)^2,
\label{eq:boundheis}
\end{equation}
showing that our bound (\ref{eq:main_thm}) is lower than the one obtained using Heisenberg's principle, as $I_E \leq 4\Delta H_C^2$. 
Note that this improvement is due to the fact that we are lowering the factor $4\Delta H_C^2$, which is the speed of evolution in Hilbert space by replacing it with $I_E$ that specifically quantifies the speed in energy space.
 
We also remark that, by using $I_E$ instead of $\Delta H_C^2$, the quantities appearing in the bound (\ref{eq:main_thm}) only depend on the battery Hamiltonian $H_B$ and the battery state $\rho(t)$. Thus, the bound is not restricted to the case of Hamiltonian evolution, but it can be used for any dynamical map. However, one should be careful with the connection between stored energy and stored work in the presence of an environment.

\subsection{Quantum advantage in power}\label{sec:Quantum_adv_power}
Once again, let us consider a battery that is made up of $N$ identical quantum cells, each with a Hamiltonian $h_j$, such that the total battery Hamiltonian reads $H_B=\sum_{j=0}^{N-1} h_j$. Now, given two charging processes, how can we meaningfully state that one of the two has a better performance in terms of power than the other? In the literature, this comparison is made in reference to the \emph{parallel charging} case \cite{1367-2630-17-7-075015, Campaioli2017, PhysRevLett.120.117702}. For a battery composed of $N$ identical quantum cells, a parallel charging process is a unitary evolution driven by a charging Hamiltonian of the form 
\begin{equation}
H_C^{||}=\sum_{j=0}^{N-1} h_c^j\, ,
\end{equation}
where the Hamiltonian $h_c^j$ locally drives the charging process of the $j$-th quantum cell in the battery. 

Now, to compare with any other unitary charging process, driven by a general charging Hamiltonian $H_C$, different quantities are chosen and normalized such that they give rise to the same scaling with the number of cells as in the parallel case. These normalization procedures impose an extensive scaling (linear in $N$) of the norm of the Hamiltonian $\norm{H_C}$, its variance $\Delta H_C^2$, or their time-averages in the case of time-dependent Hamiltonians, equal to the case with $H_C^{||}$. Under this constraint, unnormalized and normalized Hamiltonians are related by a rescaling $H_C \rightarrow x(N) H_C$.  The rescaling ensures that the total energy available to drive the charging process is always the same at order $N$. In the context of these normalization criteria, speed-ups in power compared to parallel charging have been theoretically explored in Refs.~\cite{1367-2630-17-7-075015, Campaioli2017}. In these works, entangled states or entangling operations are considered to be closely related with such speed ups.

There are two main issues with the aforementioned approach to compare among batteries. First, the presented normalization criteria may not correspond with the real experimental limitations. It could very well be that the experimentalist is limited by the strength of the local interactions but not by its amount. Hence, a fair comparison will mainly depend on the experimental capabilities. A solution to this limitation can be given by our approach, in terms of the bound in \eqref{eq:main_thm}, as it is derived irrespective of any normalization. There, the bound is given in terms of $I_E$, which may scale faster than $N$.

Once the speed of evolution in energy eigenspace has been properly taken into account, we focus on the second term in the bound \eqref{eq:main_thm}, i.e., the energy variance $\Delta(H_B)^2_\rho$, and we analyze it from the geometrical point view. In particular, this terms encodes the information about how quantumness, such as quantum entanglement, plays an important role in enhancing the power of the battery. Below, we study its relation with quantum entanglement, and relate it with the charging power.

\subsection{Relation of power with entanglement}\label{sec:fluct}
Consider a battery that is made up of $N$ identical quantum cells, with the battery Hamiltonian  $H_B=\sum_{i=0}^N h_i$. The variance of $H_B$ for an arbitrary state $\rho$ reads
\begin{equation}\label{eq:correlations_energy}
\begin{split}
\Delta(H_B)^2_\rho&= \sum_i \left(\tr (h_i^2\rho)-\tr(h_i\rho)^2\right) \\
&+\sum_{i\neq j} \left(\tr (h_i h_{j}\rho)-\tr(h_{i}\rho)\tr(h_{j}\rho)\right)\, .
\end{split}
\end{equation}

Let us discuss the relation of entanglement with \eqref{eq:correlations_energy} for the case of pure states of the battery. The first sum corresponds to the single-cell energy variance, which we denote as $ \Delta^{Loc}(H_B)^2_\rho=\sum_i \left(\tr (h_i^2\rho)-\tr(h_i\rho)^2\right)$. This quantity scales linearly with the number of cells $N$ and coincides with $\Delta(H_B)^2_\rho$ in the case of separable states ($\rho=\bigotimes \rho_i$).
 As a consequence, the only way for the variance to scale faster than $N$ is that the battery state is non-separable, and thus entangled. 
 
Indeed, in some cases one can even bound the energy variance $\Delta(H_B)^2_\rho$ with the multipartite entanglement properties of the battery state \cite{PhysRevA.85.022322, PhysRevA.85.022321}. For example, let us consider linear qubit Hamiltonians, typically used as models of a quantum battery \cite{1367-2630-17-7-075015, PhysRevLett.120.117702, PhysRevA.97.022106}, and that will be studied in Secs.~\ref{sec:ParadigmaticExamples} and \ref{sec:BoundsSpinModel}, which have the form 
\be 
\label{eq:hamiltonian_ls}
H_B = \frac{1}{2} \sum_{j=0}^{N-1} \sigma_z^j\, ,
\ee
where $\sigma_z^j$ is the $z$ Pauli matrix corresponding to the $j$-th site of the qubit chain, and we have defined the single-cell energy spacing as the unit of energy. To characterize the multipartite entanglement of this system consisting on $N$ qubits, one introduces the notion of a $k$-producible state \cite{PhysRevA.71.052302, G_hne_2005}: a pure state is $k$-producible if it is a tensor product of at most $k$-qubit states, that is
\be 
\ket{\Psi_{k-\text{prod}}}=\otimes_{l=1}^M\ket{\Psi_l}\, ,\hspace{5mm} \ket{\Psi_l} = \otimes_{j}\ket{\phi_j}\, , \hspace{5mm} \# j \leqslant k\, ,
\ee
where $\ket{\phi_j}$ is a single qubit state on the site $j$ of the chain. A mixed state is $k$-producible if it is a mixture of pure $k$-producible states
\be 
\rho_{k-\text{prod}}=\sum_\alpha p_\alpha \ket{\Psi_{k-\text{prod}}^\alpha}\bra{\Psi_{k-\text{prod}}^\alpha}\, , \hspace{5mm} \left(\sum_\alpha p_\alpha = 1\right).
\ee
Based on this classification of states, a state is $k$-qubit entangled if it is $k$-producible but not $(k-1)$-producible. For our purposes, a useful inequality for $N$-qubit $k$-producible states is \cite{PhysRevA.85.022322, PhysRevA.85.022321}  
\be\label{eq:var_ent}
4\Delta(H_B)^2_\rho \leqslant rk^2+(N-rk)^2,
\ee
where $r$ is the integer part of $N/k$. Any state (pure or mixed) that violates the bound of Eq.~\eqref{eq:var_ent} thus contains $(k+1)$-qubit entanglement. Using the above inequality, we recast the bound on power and arrive to the corollary \eqref{cor:PowerEnt}.

\begin{corollary}[Power and entanglement]  \label{cor:PowerEnt}
Given a process for charging (or discharging) a quantum battery composed of $N$-qubits, with a battery Hamiltonian of the form of Eq.~\eqref{eq:hamiltonian_ls}, if at time $t$ the battery is at most $k$-qubit entangled, its instantaneous power fulfills 
\begin{equation}\label{eq:main_thm_ent}
P(t)^2\leqslant \Delta H_B(t)^2  I_E(t) \leqslant  \underbrace{\frac{1}{4}\left[rk^2+(N-rk)^2\right]}_{\text{$k$-\normalfont{producibility}}} \,  \underbrace{I_E(t)}_{\substack{\text{\normalfont{disting. in}} \\ \text{\normalfont{energy eig.}}}},
\end{equation}
where $r$ is the integer part of $N/k$. In the cases where $r$ is an exact integer, the inequality reduces to 
\begin{equation}\label{eq:main_thm_ent2}
P(t)^2\leqslant \frac{kN}{4}\,  I_E(t).
\end{equation}
\end{corollary}
Notice that the bound of Eq.~\eqref{eq:main_thm_ent} sets the limitations imposed in power by the multipartite entanglement properties of the battery state and its distinguishability in the energy eigenspace, both for pure and mixed states. Thus, this bound leads to a deeper understanding of the relation between entanglement and power in qubit-based quantum batteries. 
For instance, the inequality \eqref{eq:var_ent} is saturated only by a battery state which is a product of $r$ Greenberger-Horne-Zeilinger  (GHZ) states of $k$ qubits $\ket{\text{GHZ}_k}$, and another set of GHZ states of $N-rk$ qubits $\ket{\text{GHZ}_{N-rk}}$\cite{PhysRevA.85.022321}, with
\be 
\ket{\text{GHZ}_k}=\frac{1}{\sqrt{2}}\left(\ket{0}^{\otimes k}+\ket{1}^{\otimes k}\right).
\ee
An example of such a GHZ state was indeed introduced (with $k=N$) in the seminal paper 
\cite{1367-2630-17-7-075015}, introducing the potential entanglement boost in quantum batteries. 

Another important insight is that many-body charging Hamiltonians with a large participation number and low order of interactions will generically lead to battery states that strongly differ from the optimal GHZ ones. In that sense, we would also like to note that it is interesting to compare our bound of Eq.~\eqref{eq:main_thm_ent} with the ones derived in the work \cite{Campaioli2017}, where the relation of the quantum advantage in Power with the interaction order and participation number has been considered.

\section{Paradigmatic examples}
\label{sec:ParadigmaticExamples}
One example that has become paradigmatic in the field \cite{1367-2630-17-7-075015, PhysRevLett.120.117702, PhysRevA.97.022106} is the charging of a battery composed of non-interacting cells by a time-independent Hamiltonian evolution. Once again, let us consider the battery Hamiltonian of $N$ two-level systems with the Hamiltonian $H_B=\frac{1}{2}\sum_{j=0}^{N-1} \sigma_z^j$, which has the same form as Eq.~\eqref{eq:hamiltonian_ls}. Initially, at $t=0$, the battery is in its ground state $\ket{\psi_0}=\ket{0}^{\otimes N}$.  Notice that the ground state has a negative energy and, as we are interested in the energy difference, hereafter we will define the stored energy at a given time as  
\be 
E(t) := \Tr(\proj{\psi(t)} H_B)-\Tr(\proj{\psi_0}H_B).
\ee
As the initial state is a pure state, which has zero entropy, one trivially gets that the capacity is given by $C_{S=0}=N$. In order to charge the battery, one can use different charging Hamiltonians $H_C$. In particular, illustrative examples are the cases of a \emph{parallel}, \emph{global} and \emph{hybrid} Hamiltonians, represented by $H_C^{||},\ H_C^{\#}$, and $H_C^{h}$, respectively. These Hamiltonians are given by
\begin{equation}\label{eq:paradigmatic_hamiltonians}
\begin{split}
H_C^{||}&=\lambda\sum_{j=0}^{N-1} \sigma_x^j, \\
H_C^{\#}&=\lambda\otimes_{j=0}^{N-1} \sigma_x^j, \\
H_C^{h}&=\lambda\sum_{j=0}^{q-1}\otimes_{i=1}^r\sigma_x^{qj+i}\, , 
\end{split}
\end{equation}
where $\sigma_x^j$ is the $x$ Pauli matrix acting on the $j$th cell, and in the \emph{hybrid} case $N=qr$.  Regarding time-units, here $\lambda$ represents a charging frequency, where we use the convention $\hbar=1$. The main features of these three charging processes are outlined in the  \tabref{table:comparison_Hamiltonians} below. 

\begin{table}[h]
\caption{Comparison between three different charging Hamiltonians of a battery composed of $N$ qubit cells: the \emph{parallel} charging, the fully interactive \emph{global} Hamiltonian (optimal) and an \emph{hybrid} construction where $m$ blocks of $q$ qubits are in parallel charged in a fully interactive way. Also see Fig.~\ref{fig:Glob_par}.}
\begin{center}
\begin{tabular}{cccc}
& {\bf Parallel}  & {\bf Global} & {\bf Hybrid}\\
\hline \\[-.5em]
$H_C$  & \hspace{.005cm} $\lambda\sum_j \sigma_x^j$ \hspace{.005cm} & \hspace{.005cm} $\lambda\otimes_j \sigma_x^j$ \hspace{.005cm} & \hspace{.005cm} $\lambda\sum_{j=0}^{q-1}\otimes_{i=1}^r\sigma_x^{qj+i}$\hspace{.005cm} \\[.5em]
\hline \\[-.5em]
$\norm{H_C}$  & $ N\lambda$ & $\lambda$ & $q\lambda $ \\[.5em]
\hline \\[-.5em]
$\lambda t_F$  & $\pi/2$ & $\pi/2$ & $\pi/2$ \\[.5em]
\hline \\[-.5em]
$\Delta H_C^2$  & $N\lambda^2$ & $\lambda^2$ & $q\lambda^2$ \\[.5em]
\hline \\[-.5em]
$I_E$  & $4N\lambda^2$ & $4\lambda^2$ & $4q\lambda^2$ \\[.5em]
\hline \\[-.5em]
$E(t)$  & $N p$ & $Np$ & $Np$ \\[.5em]
\hline \\[-.5em]
$\Delta H_B(t)^2$  & $N p(1-p)$ & $N^2 p(1-p)$ & $Nr p(1-p)$ \\[.5em]
\hline  \\[-.5em]
$k$-$qubit\: ent.$ & $1$ & $N$ & $r$ \\[.5em]
\hline \\[-.5em]
$\Delta^{Ent} H_B(t)^2$  & $0$ & $Np(N-1)(1-p)$ & $Np(r-1)(1-p)$ \\[.5em]
\hline  \\[-.5em]
$P(t)$  & $\sqrt{\Delta^2H_BI_E}$ & $\sqrt{\Delta^2H_BI_E}$ & $\sqrt{\Delta^2H_BI_E}$ \\[.5em]
\hline  \\[-.5em]
\end{tabular}
\vspace{-4mm}
\end{center}
\begin{flushleft}
*$p=\sin^2(\lambda t)$.\\
**Note, at $t_0$ and $t_f$ the state of the system is always a product state regardless of the form of $H_C$.
\end{flushleft}
 \label{table:comparison_Hamiltonians}
\end{table}

It is easy to see that all these Hamiltonians evolve the initial ground-state state to the highest energy state $\ket{1}^{\otimes N}$. Therefore, the stored energy coincides with the capacity at time $\lambda t_f=\pi/2$, as per Observation \ref{thm:capacity}. Power is the same for all the cases, and the bound \eqref{eq:main_thm} is always saturated. Furthermore, for these cases, the evolution speed in state space and energy eigenspace coincide, as $I_E = 4\Delta H_C^2$ (see \tabref{table:comparison_Hamiltonians}). 
 As can be seen in Fig.~\ref{fig:Glob_par}, without any normalization constraints, $I_E$ scales linearly with $N$ in the \textit{parallel} case, as there are many energy levels involved in the evolution. The situation is drastically different in the \textit{global} case, where only the groundstate and maximally excited energy levels participate, leading to an $N$-independent scaling of $I_E$. The \textit{hybrid} case presents an intermediate behavior.
 
The quantum enhancement in power in these paradigmatic examples was understood \cite{1367-2630-17-7-075015,Campaioli2017} with the help of certain normalization criteria, as explained in \secref{sec:Quantum_adv_power}. The approach in \cite{Campaioli2017}, for instance, imposes a linear scaling of $\Delta H_C^2$ with $N$, and one should perform the normalization $H_C^{\#}\rightarrow \sqrt{N}H_C^{\#}$. With this criterion, the global charging is $\sqrt{N}$ times faster than the parallel one. On the other hand, under the constraint of $\norm{H_C}$ scaling as $N$, imposed in Ref.~\cite{1367-2630-17-7-075015}, one should perform the normalization $H_C^{\#}\rightarrow {N}H_C^{\#}$. Within such norm rescaling, the global charging performs $N$ times faster than the parallel one.
 
Instead, the approach presented in the previous sections allows us to discuss the role of the entanglement between the cells in the power of these models, without the need of relying to a normalization constrain, and in a quantitative manner. It is easy to see that the \textit{global} (\textit{r-hyrbid}) Hamiltonians leads to a battery state that will have $N$($r$)-qubit entanglement:
\be 
\begin{split}
\ket{\Psi}_\text{\textit{global}}=&\sqrt{1-p} \ket{0}^{\otimes N}+\sqrt{p}\ket{1}^{\otimes N},\\
\ket{\Psi}_\text{\textit{hybrid}}=&\otimes_{j=0}^{q-1}\ket{\Psi(t)}_j,\\
 \ket{\Psi}_j=&\sqrt{1-p} \otimes_{i=1}^r \ket{0}_{qj+i}+\sqrt{p}\otimes_{i=1}^r\ket{1}_{qj+i}.
\end{split}
\ee
Notice (compare rows 7 ad 8 of Table \ref{table:comparison_Hamiltonians}) that this amount of $k$-qubit entanglement impacts the super-linear scaling factor of $\Delta^2H_B$, and thus also the one of power. In the middle of the evolution, for $t=\pi/4$ the states are exactly in the GHZ form, and the bound of Eq.~\eqref{eq:main_thm_ent} is saturated.

\begin{figure}
\includegraphics[width=1\linewidth]{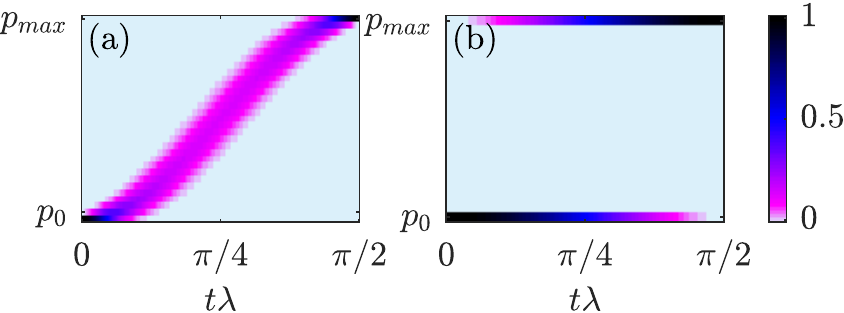}
\caption{Dynamics of energy levels $p_k$, according to their definition (\ref{eq:def_pk}), during   a parallel (a) and a global (b) charging process. Here one can see the contribution of each energy level (y-axis) to the total energy, as a function of time. One can see absolute non-locality in energy space for the global case, where there is a coexistence in time of the highest and lowest energy level,  and locality of the parallel one, where at a given time only levels that are close contribute to the total energy. }
\label{fig:Glob_par}
\end{figure}

\section{Specific spin models \label{sec:BoundsSpinModel}}

In this section, we study our previously derived bounds in specific spin-$\frac{1}{2}$ models, effectively described in terms of qubits. The models, which can in principle be realized experimentally (see, for instance, Ref.~\cite{MaciejBook}), are (i) integrable spin models in 1D with ultracold fermionic atoms, (ii) the Lipkin-Meshkov-Glick (LMG) model with ultracold atoms or atoms near nanostructures, and (iii) Dicke model with ultracold ions, BEC in an optical cavity, or cavity circuit QED. While it is true that some of these spin models have already been presented in the literature \cite{1367-2630-17-7-075015, PhysRevLett.120.117702, PhysRevA.97.022106} as candidates for an experimentally realizable quantum battery, we use our formalism to systematically analyze them. More specifically, for each model we discuss the impact of both the evolution speed, quantified by $I_E$, and the smartness of the path undergone in the power of the battery, related to $\Delta H_B^2$ and $k$-qubit(spin) entanglement, hence clarifying the origin of possible speed-ups. Furthermore, we also study the amount and quality of stored energy in the final battery state.

In all the cases we consider that the quantum battery is a chain of $N$ spin-$\frac{1}{2}$ cells.  We work in the local basis of Pauli matrices for each spin, i.e., any local operator acting on $j$-th spin can be expressed in terms of \{$\sigma_x^j,\sigma_y^j,\sigma_z^j,\mathbb{I}$\}, and define our battery Hamiltonian as
\be \label{eq:SpinBattery}
H_B=\frac{1}{2}\sum_{j=0}^{N-1}\sigma_z^j.
\ee
We consider that the spin system is initially in the ground state of the battery Hamiltonian $H_B$. We now look at different charging models.

\subsection{Integrable spin models}

Here we consider a general class of charging Hamiltonians in 1D of the form
\be \label{eq:hamiltonianSpin}
 \begin{split}
 H_{JW} &= H_B +\frac{1}{2}\sum_{\substack{j=0 \\ m=1}}^{N-1}\left[(\lambda_m+\gamma_m)\sigma^{j}_x(\otimes_{l=j+1}^{j+m-1}\sigma^l_z)\sigma^{j+m}_x+\right.\\
 &+\left.(\lambda_m-\gamma_m)\sigma^{j}_y(\otimes_{l=j+1}^{j+m-1}\sigma^l_z)\sigma^{j+m}_y\right],
 \end{split}
 \ee
that can be diagonalized exploiting the Jordan-Wigner (JW) transformation. Above, we have implicitly assumed translational invariance and periodic boundary conditions. The above family of Hamiltonians includes the 1D transverse field Ising model and XY model with a transverse field if we limit the interaction range to nearest-neighbors only. The dynamics of these spin systems, parametrized by ($\lambda_m,\gamma_m$), can be easily solved (see Appendix~\ref{App:A}) by a mapping through the aforementioned  JW transformation to a fermionic chain, followed by a Fourier transformation of the fermionic operators exploiting the translational invariance, and a final Bogoliubov transformation of the Fourier transformed fermionic operators. 

What is important in our discussion of these systems, acting as quantum batteries, is that in the fermionic picture they present a local structure in momentum space due to their translational invariance. This means that the Hilbert space structure of the problem can be expressed as
\be 
\mathcal{H}=\bigotimes\mathcal{H}_{k,-k},\label{eq:hilbertstructure}
\ee
where $k$ labels the quasimomentum, and thus these models are very similar to an \textit{hybrid} model with $r=2$ (see Table \ref{table:comparison_Hamiltonians}), which is close to the \textit{parallel} case, as there is only two-particle entanglement between $(k,-k)$ modes.  We see this fact in that we are able to write all the relevant dynamical quantities (see Appendix~\ref{App:A}) as a sum of independent contributions from each $(k,-k)$ subspace. For instance,
\be 
\begin{split}
E(t)&=\sum_{k\in \# \text{BZ}} \varepsilon_k(t),\\
P(t) &= \sum_{k\in \# \text{BZ}} \dot{\varepsilon}_k(t),\\
\Delta H_B(t)^2&=\sum_{k\in \# \text{BZ}} {\varepsilon}_k(t)(2-{\varepsilon}_k(t)),\\
\Delta H_{JW}^2&=\sum_{k\in \# \text{BZ}} \sin^2(\theta_k)\omega^2_k,
\label{eq:jordan}
\end{split}
\ee
where $\varepsilon_k$, $\theta_k$, and $\omega_k$ depend on the parameters of the model and $\#BZ$ refers to the reduced Brillouin zone of the fermionic chain (see Appendix~\ref{App:A}). Notice that, as the size of the reduced Brillouin zone is proportional to $N$, the quantities appearing in \eqnsref{eq:jordan} have a natural linear scaling with $N$. 

Comparing the \textit{parallel} and \textit{hybrid} rows of Table ~\ref{table:comparison_Hamiltonians}, we observe that these integrables models can present a quantum advantage by at most a factor of 2, under the criteria that the speed factor $I_E$ is kept at a constant value. This result is in agreement with the result of Ref.~\cite{Campaioli2017}, where it was shown that, for a fixed $\Delta H_{JW}^2$, the maximum quantum advantage is given by the order of interaction of $H_{JW}$.

\begin{figure}
  \includegraphics[width=1\linewidth]{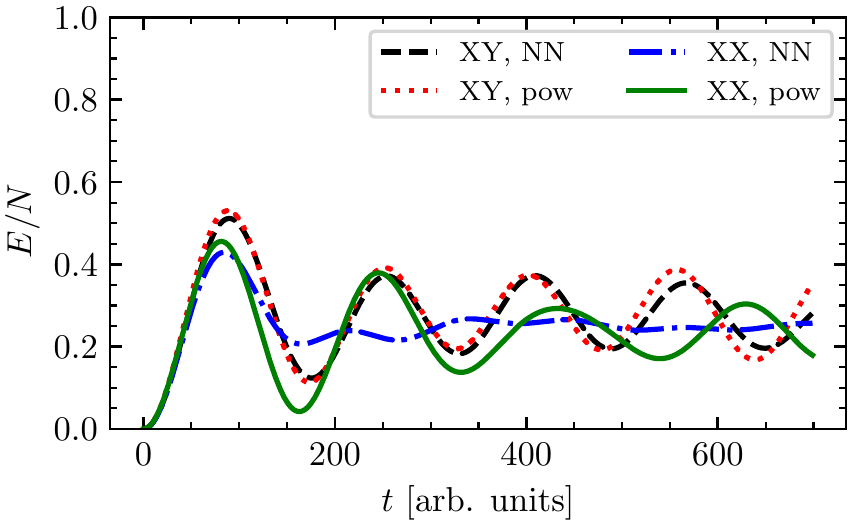}
\caption{Dynamics of the stored energy for different spin integrable models described by the Hamiltonian \eqnref{eq:hamiltonianSpin}. The legend indicates to which model corresponds each line. We have studied the XY (XX) model that correspond to  $\lambda_m=0$ ($\lambda_m=\gamma_m$). For each of these models we have studied the nearest-neighbor (NN) case, where $\gamma_1=1$ and $\gamma_{m\neq 1}=0$, and also the power law (pow) case, where $\gamma_m = m^{-2}$. We have fixed the system size to $N=20$ spins and normalized the stored energy (Y-axis) accordingly, such that it is bounded to 1. We observe that for all the models the maximum stored energy is about $50\%$ of the total due to the desynchronization between the different modes that contribute to this quantity.}
\label{fig:dynamicsfermions}
\end{figure}

 \begin{figure}
  \includegraphics[width=\linewidth]{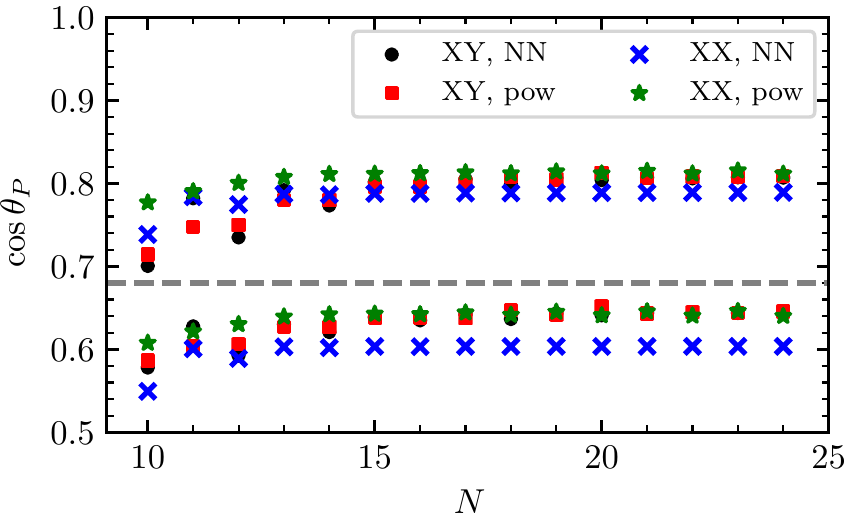}
\caption{Over the grey dashed line, we plot the value of $\cos{\theta_P}\coloneqq P/{\sqrt{\langle\Delta H_B^2\rangle_{\Delta t} \langle I_E\rangle_{\Delta t}}}$ as a function of the number of spins $N$ for the models presented in \figref{fig:dynamicsfermions}. We observe a fast saturation to an approximate 0.8 value. Below the grey dashed line, we plot the same quantity but substituting $\langle I_E\rangle_{\Delta t}\rightarrow 4\Delta^2H_{JW}$, and it saturates to an approximate 0.6 value. This shows how lower is our bound obtained using the Fisher information in energy eigenspace instead of the speed in Hilbert space for these particular models.}
\label{fig:costhetapfermions}
\end{figure}

The only collective effect we are left with is the synchronization between independent modes [two producible states of $(k,-k)$ particles]. This phenomenon has a direct consequence: it limits the capacity and quality of stored energy, as in general the set of energies $\varepsilon_k(t)$ will not be maximized simultaneously. This feature is discussed in \figref{fig:dynamicsfermions} for particular choices of $(\lambda_m,\gamma_m)$. It is also important to notice that, in general, the system will not be in an energy eigenstate when it reaches maximum capacity. However, the energy uncertainty associated with the final state will be negligible in the thermodynamical limit, as $\Delta H_B(t_f)/E(t_f) \sim 1/\sqrt{N}$.

Moreover, one can compute the Fisher information in energy eigenspace for these models (see Appendix.~\ref{App:B}) to see how tight our bound is for different choices of parameters and as a function of $N$, a result which is shown in \figref{fig:costhetapfermions}.

\subsection{Lipkin-Meshkov-Glick model}

Another class of charging Hamiltonians that we consider is based on the LMG model \cite{1965NucPh..62..188L}, which allows for two-body spin interactions with an infinite range. Namely, the charging Hamiltonian is given by
\begin{equation}\label{eq:LMG_ham_local}
     H_{LMG}=  \frac{\lambda}{N}\sum_{i<j}(\sigma_x^i\sigma_x^j+\gamma\sigma_y^i\sigma_y^j ) + \frac{1}{2}\sum_{j=0}^{N-1}\sigma_z^j,
\end{equation}
where $\lambda$ is the coupling strength, $\gamma$ the anisotropy parameter, and the factor $\frac{1}{N}$ is included in the model in order to have a finite interaction energy per spin in the thermodynamic limit. For the infinite range Ising model ($\gamma = 0$), $H_{LMG}$ is analog to the so-called twist-and-turn Hamiltonian
\cite{PhysRevA.92.023603}. There $\lambda$ mimics the twisting parameter, and the linear term coming from $H_B$ is a rotation around the z-axis. Using the components of the total spin operator $\boldsymbol{J}\coloneqq (J_x,\ J_y,\ J_z)$, with
\be 
J_\alpha = \sum_{j=0}^{N-1} \frac{\sigma_\alpha^j}{2}, \ \ \ \ (\alpha = x,\ y,\ z),
\ee 
the LMG Hamiltonian of \eqnref{eq:LMG_ham_local} can be rewritten as
\begin{equation}\label{eq:LMG_ham_total}
H_{LMG} = \frac{\lambda}{2N}\Big{[}(1+\gamma)\Big{(}J_+J_-+J_-J_+-N\Big{)}+(1-\gamma)\Big{(}J_+^2+J_-^2\Big{)}\Big{]} + J_z,
\end{equation}
where we have introduced the ladder operators $J_+$ and $J_-$, that are related to the total spin operators by $J_x $=$ \frac{1}{2}(J_+$+$J_-)$, and $J_y$=$\frac{1}{2i}(J_+$-$J_-)$.  Notice that, in the total spin notation, the battery Hamiltonian reads $H_B=J_z$. Note also that $\boldsymbol{J}$ is a constant of motion, i.e., $[\textbf{J},H_{LMG}]=0$, and that the initial state lies in the maximum spin sector. These properties effectively reduce the size of the Hilbert space of the problem, that scales only linearly with $N$, instead of the exponential $2^N$ scaling of the total Hilbert space described by local Pauli matrices.

\begin{figure}[h!]
\includegraphics[width=1\linewidth]{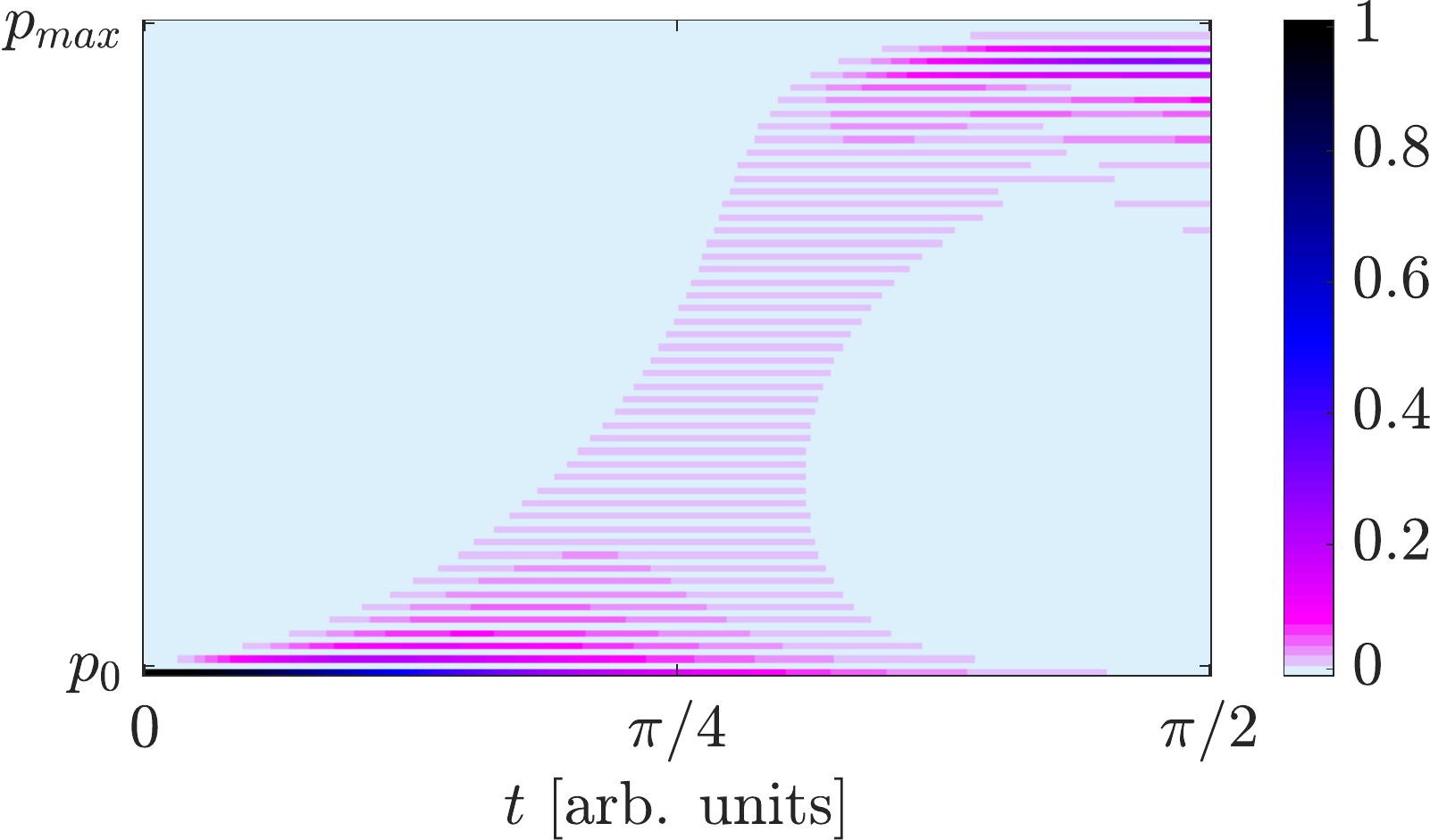}
\caption{Dynamics of the energy levels during the charging process based on the LMG model for $\lambda = 5$. Due to the structure of $H_{LMG}$, only every second energy level starting from the initial state is occupied during the evolution.}
\label{fig:lmg}
\end{figure}

Let us first discuss some general properties of the Hamiltonian \eqref{eq:LMG_ham_total} that mainly affect the capacity properties of the LMG model. First, note that deposition of energy into the battery only occurs if $\gamma \neq 1$, and the maximum capacity is achieved for $\gamma=-1$ [see Figs.~\ref{fig:LMG_scaling_C}(c) and ~\ref{fig:LMG_scaling_C}(d)]. This comes from the fact that the mixed terms in the Hamiltonian (i.e. $J_+J_-$ and $J_-J_+$) are diagonal operators in the energy eigenbasis, and hence they only contribute to the free evolution and not to the charging process.  

Second, there are two regimes depending on the value of the parameter $\lambda$. The strong-coupling regime is defined by $\lambda \geq 1$, whereas in the weak-coupling regime $\lambda < 1$. In the latter, the LMG model leads to very poor charging properties, as the maximum stored energy tends to zero in the thermodynamic limit. This is due to the fact that the ground state of the battery Hamiltonian $H_B$ (i.e., the initial state) is also an eigenstate of $H_{LMG}$, when $N\to \infty$ \cite{PhysRevA.69.022107}. Therefore, for our discussion, we will focus on the strong coupling regime.

In studying power for LMG model, a first quantity of interest is the variance of the charging Hamiltonian
\begin{equation}
\Delta H_{LMG}^2 = \frac{\lambda^2}{2}(1-\gamma)^2 + \textit{O}\Big{(}\frac{1}{N}\Big{)}.
\label{eq:LMG_lim}
\end{equation} 
We see that these variance does not scale with $N$, implying that $I_E$ cannot scale with $N$ either, as $I_E \leq 4\Delta H_{LMG}^2$. Hence, our bound (\ref{eq:main_thm}) tells us that any scaling of the power with $N$ can only be associated with $\Delta H_B^2$. Note, the correlations are expected to be enhanced in this model due to the long-range nature of the interactions between spins, and also because it does not have any ``hidden'' local structure in energy, unlike the previous case of integrable spin models. Nevertheless, they are bounded to scale as $N^2$ at most (see \secref{sec:fluct}). As a consequence, power can scale at maximum linearly with $N$ for the LMG charging Hamiltonian, and no $N$-dependent speed-up is possible.

\begin{figure}
\includegraphics[width=1\linewidth]{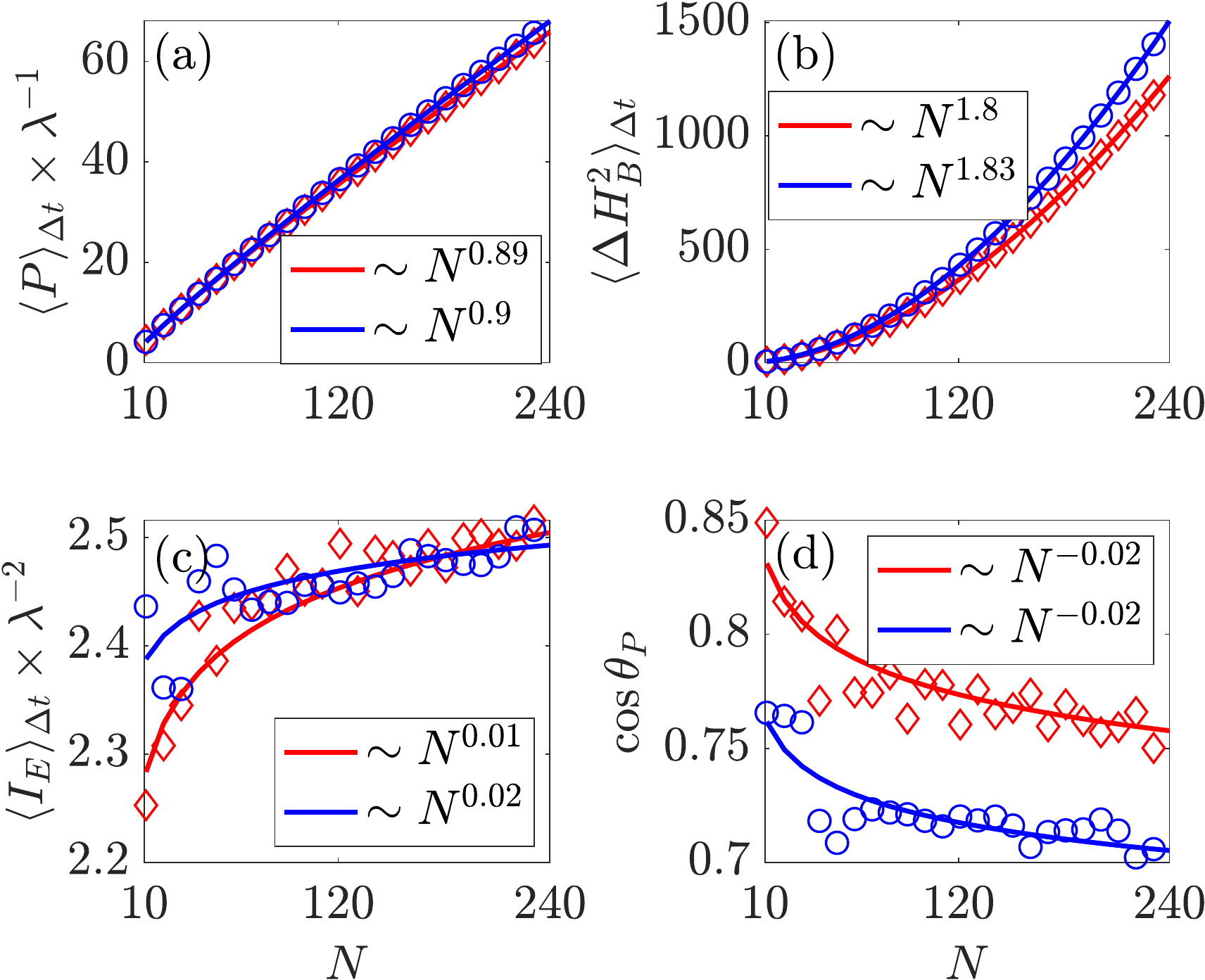}
\caption{The figures consider time-averaged quantities of the LMG model, which are relevant in the study of power, as a function of the number of spins $N$. The final time has been chosen as the one for which the capacity is maximized. Blue color is used for $\lambda=20$, whereas red color is used for $\lambda=5$. Quantities that carry units of time (i.e., power and Fisher information) have been renormalized with the coupling strength $\lambda$. The legends indicate the scaling with $N$ of the different curves.}
\label{fig:scaling_LMG}
\end{figure}

To make a quantitative analysis of the scaling of such correlations with the number of spins $N$, and also study the tightness of our bound on power, we have solved its dynamics for two values of the coupling strength in the strong-coupling regime $(\lambda = 5,\ 20$), and a fixed value of $\gamma=-1$.

First, we would like to draw attention to \figref{fig:lmg}, where one can visualize the evolution of the LMG battery in the eigenspace of $H_B$. If one compares this evolution with the ones of the paradigmatic cases (see \figref{fig:Glob_par}), one observes that the LMG battery has some common properties with the global charging case, as there appears entanglement between the states that are far away in energy during evolution. This enhances the battery variance $\Delta H_B^2$, as explained below. However, in the LMG battery, there are many energy levels involved in the charging process, contrary to what happened in the global charging case, where only the lowest and highest energy states participate.

In \figref{fig:scaling_LMG}(b), one can see the enhancement of the time-averaged energy variance $\Delta H_B^2$ in the LMG battery. Even though this variance does not saturate the bound of $N^2$ scaling, it definitely scales super-extensively, as $\sim N^{1.8}$. As per Eq.~\eqref{eq:var_ent}, this scaling means that the battery is in a highly $k$-qubit entangled state during the charging process. In \figref{fig:scaling_LMG}(c) one observes that the Fisher information $I_E$ does not vary with the system size, a result that is in agreement with the analytical formula of \eqnref{eq:LMG_lim}. 
Finally, we note that, for the LMG battery in the strong-coupling regime, the bound (\ref{eq:main_thm}) is tight at order $N$ (see \figref{fig:scaling_LMG}d), and power scales approximately linearly with $N$ [see \figref{fig:scaling_LMG}(a)], inheriting the scaling of $\Delta H_B$. We also remark that increasing the driving parameter $\lambda$ decreases the tightness of the bound.

\begin{figure}[t]
  \includegraphics[width=1\linewidth]{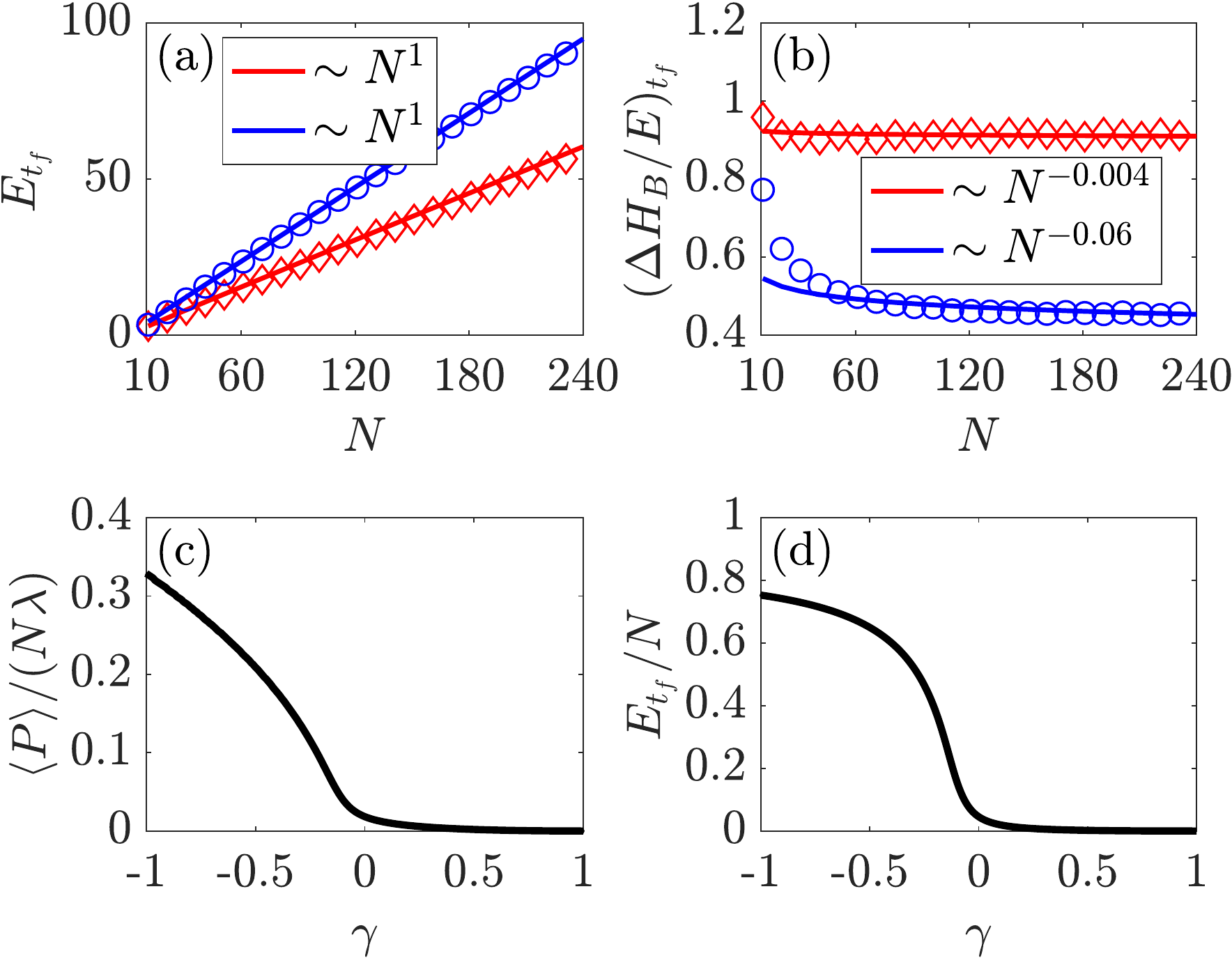}
    \caption{Capacity properties of the LMG battery. In figure (a) and (b) one can see the expected scaling of the capacity and the final relative energy variance as a function of the system size. The blue color is used for $\lambda=20$, whereas red color is used for $\lambda=5$, and the legends indicate the scaling law. In (c) we plot the dependence with the anisotropy parameter $\gamma$  of both the capacity and power. For this plot, we used a system of $50$ spins and set $\lambda=5$.}
\label{fig:LMG_scaling_C}
\end{figure}

At this point, it is interesting to compare the conclusions obtained for the LMG model within the scope of the bound (\ref{eq:main_thm}) and the analysis on speed-ups under certain normalization criteria. In \cite{Campaioli2017}, it was predicted that, under a fixed linear scaling with $N$ of $\Delta H_C^2$, power could scale super-extensively in batteries with a large participation number, e.g. the present LMG model. In \secref{sec:Quantum_adv_power}, we showed that if one imposes a linear scaling with $N$ of the Fisher information $I_E$ (and thus of $\Delta H_C^2$), speed ups in power are directly related to the enhancement of the battery variance $\Delta H_B^2$. We have seen that, in the LMG battery, such variance is indeed highly enhanced and it translates into power, as the bound (\ref{eq:main_thm}) is tight for this model. However, this enhancement in power does not come from the speed of evolution, as the Fisher information $I_E$ remain invariant with the system size. Thus the normalization criterion under which such enhancement in power was predicted does not apply.

Hence, the example of the LMG model shows the importance of both $\Delta H_B^2$ and $I_E$, and it stresses the two-fold origin of quantum advantage in power.

To conclude with the LMG model, let us now briefly discuss the energy capacity and its variance, quantitatively.  In \figref{fig:LMG_scaling_C}a, we see that in the strong-coupling regime the energy stored scales linearly with the number of cells $N$, as expected. However, in \figref{fig:LMG_scaling_C}b we see that there is no decay of the relative variance of the final stored energy in the large $N$ limit. This presents an important problem for a deterministic extraction of the stored energy, and it means that the super extensive energy variance $\Delta H_B^2$ that build-up during the charging process to enhance power, does not disappear in the final state.

\subsection{Dicke model}\label{sec:Dicke}

A quantum battery can also be constructed by placing an array of spins inside an optical cavity \cite{DickePhysRev.93.99}. This particular model has been studied in Ref.~\cite{PhysRevLett.120.117702}, where a collection of $N$ spins interact with a cavity field mode. In this section, we reconsider this battery model and study it in the light of our bounds presented before. The paradigm assumed here is qualitatively different from the previous ones, as the system that provides the energy to the battery  (i.e., the charging agent) is explicitly considered. The battery Hamiltonian is still given by $H_B = J_z$. The charging Hamiltonian includes the free evolution of both the spins and the cavity and a linear interaction between them. It reads
\be \label{eq:dickehamiltonian}
H_{DK} = J_z+\adop\aop+\frac{2\lambda}{\sqrt{N}}J_x(\adop+\aop), 
\ee
where $\adop (\aop)$ are the usual creation (annihilation) operators of cavity photons, satisfying $[\aop,\, \adop]=\mathbb{I}$, and the macro-spin notation is adopted as previously. In contrast to the convention used in \cite{PhysRevLett.120.117702}, here we include the factor $\frac{1}{\sqrt{N}}$ in the coupling in order to have a well defined thermodynamical limit \footnote{The thermodynamical limit is considered on the cavity-spin system as a whole. This implies that if one adds spins the cavity length $L$ should also increase to allocate them, keeping the density $N/L$ constant, or in other words $N\propto L$. On the other hand, the coupling term in the Hamiltonian (\ref{eq:dickehamiltonian}) has its origin in the electric-dipole interaction $\sum E_j \cdot d_j$, where $d_j$ is the dipole of each spin, and $E_j$ the electric field at the spin position. Of course, the individual dipole operators $d_j$ do not carry any scaling with $N$, and they are simply proportional to $\sigma^j_x$, but the electric field quantized inside a cavity of length $L$ carries a normalization factor $1/\sqrt{L}$. Thus $L$ and $N$ need to be proportional.}, for $N\to\infty$. We consider that in the initial state the spins (i.e., the battery) are in the ground state of $H_B$, and the cavity is in the eigenstate of the photon number that contains $N$ photons. 

A first observation is that one can analytically compute
\be \label{eq:dickefluct}
\Delta H_{DK}^2=2\lambda^2(2N+1)\, ,
\ee
and therefore one sees that $I_E$ is bounded to scale linearly with $N$ (parallel case scaling) at most. One is then left with $\Delta H_B^2$ as the only quantity appearing in our bound (\ref{eq:main_thm}) that could scale super-extensively in order to enhance the power scaling.  Notice that this is only possible if a normalized coupling $\lambda/\sqrt{N}$ is introduced.
The non-normalized case, considered in Ref.~\cite{PhysRevLett.120.117702},  can be experimentally engineered \cite{PhysRevLett.103.083601} if one considers that $N$ ranges from $1$ to some large but finite value. In such scenario, we have $\Delta \tilde{H}_{DK}^2=2N\lambda^2(2N+1)$. 
Then, the power can scale super-extensively without the need of a super-extensive scaling in of $\Delta H_B^2$, i.e., the spins do not need to explore highly correlated subspaces.

\begin{figure}[h!]
\includegraphics[width=1\linewidth]{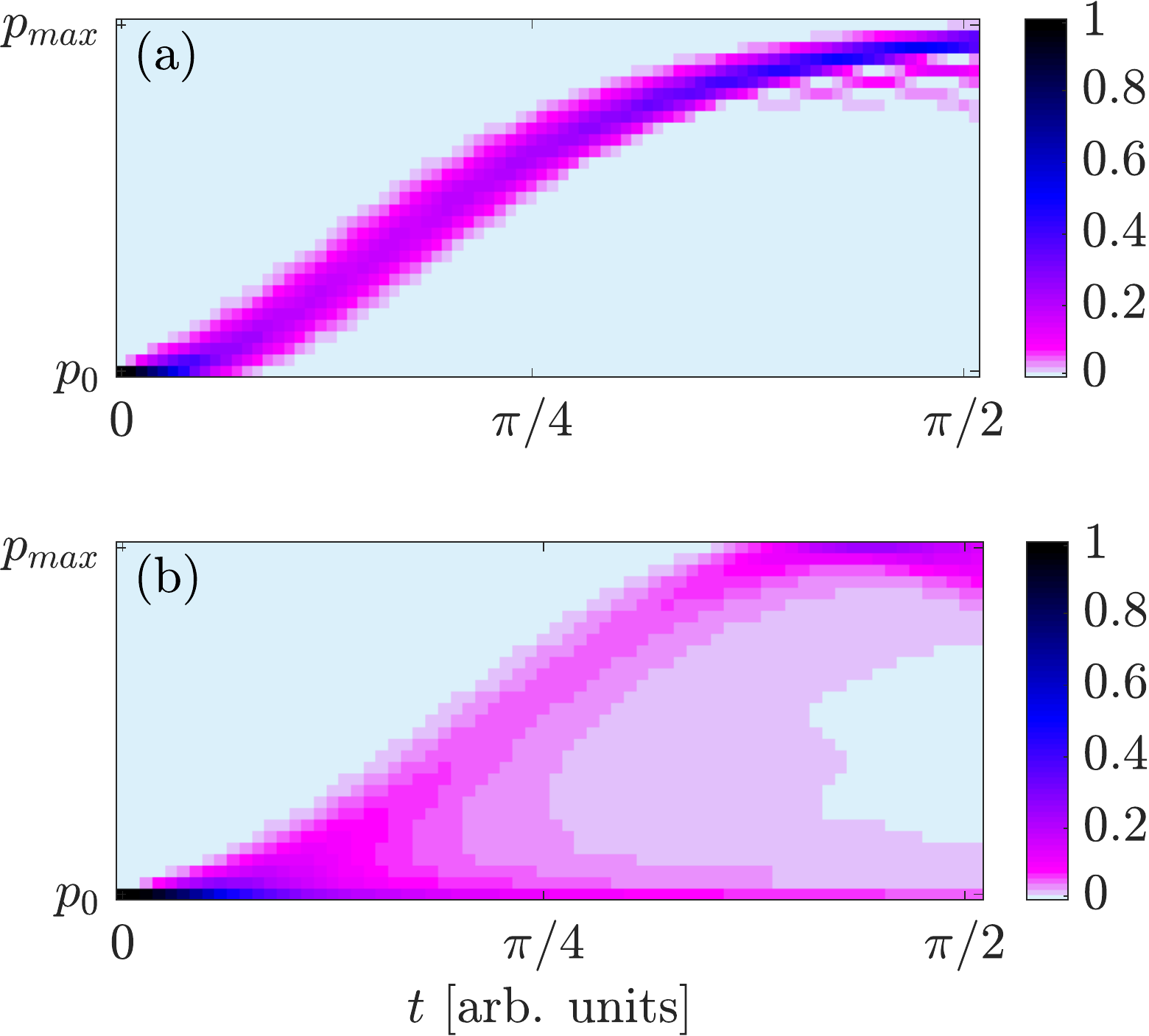}
\caption{Dynamics of the energy levels $p_k$ during the charging process based on Dicke model for $\lambda$=$0.01$ (a) and $\lambda$=$0.5$ (b). In (a) we see that for the weak-coupling regime the behavior is similar to the parallel charging case. On the other hand, (b) shows an intermediate behavior between parallel and global charging.}
\label{fig:dickie}
\end{figure}

\begin{figure}[h!]
\includegraphics[width=\columnwidth]{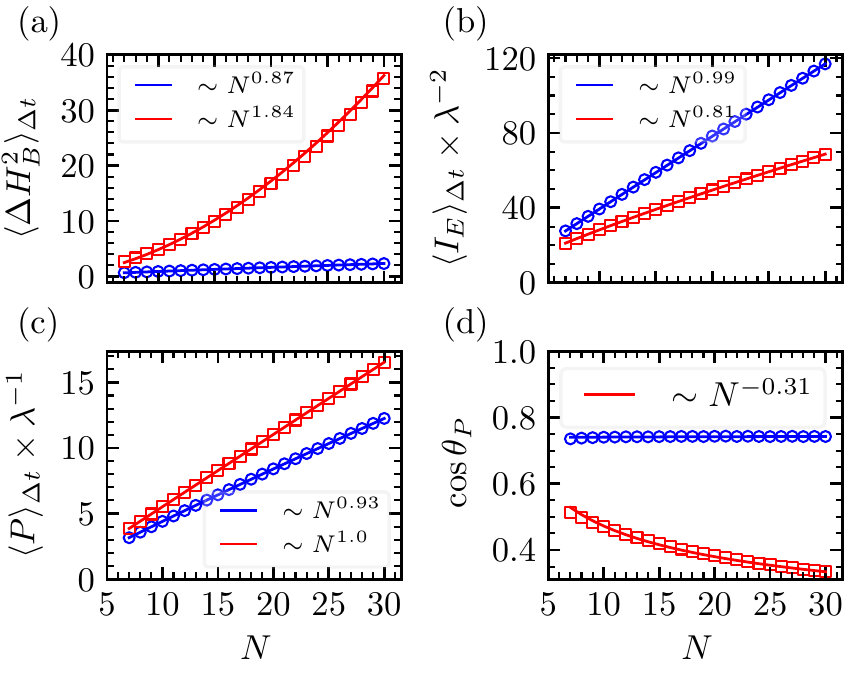}
\caption{The figures consider time-averaged quantities of the Dicke model, which are relevant in the study of power, as a function of the number of spins $N$ inside the cavity. The final time has been chosen as the one for which the capacity is maximized. The blue color (circles) is used for the weak-coupling regime ($\lambda=0.01$), whereas the red color (squares) is used for the strong-coupling regime ($\lambda=0.5$). Quantities that carry units of time (i.e., power and Fisher information) have been renormalized with the coupling strength $\lambda$. The legends indicate the scaling with $N$ of the different curves.}
\label{fig:Dicke2}
\end{figure}

\begin{figure}[h!]
\includegraphics[width=\columnwidth]{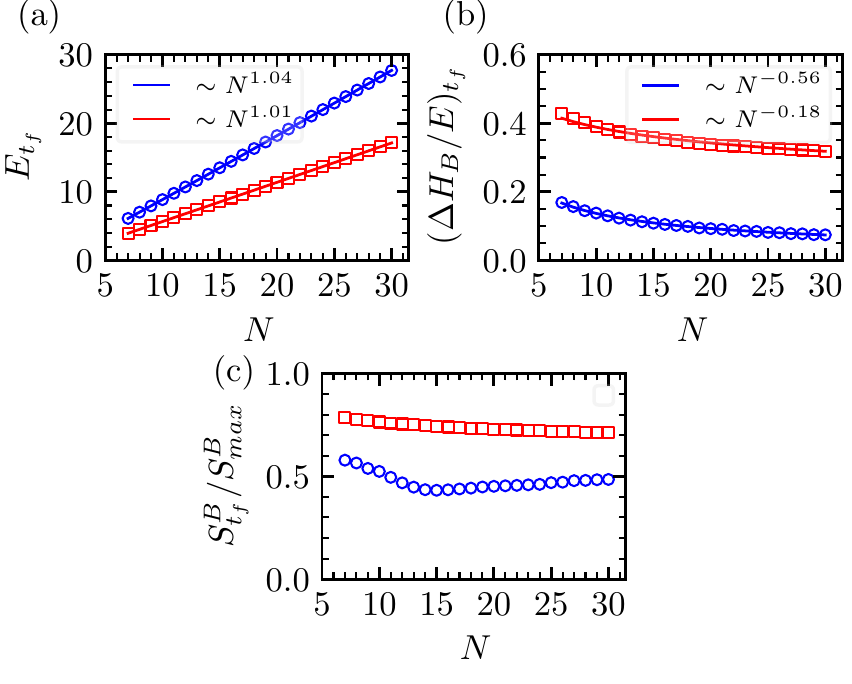}
\caption{The figures show the relevant quantities for the capacity of the Dicke model evaluated at the final time $t_f$, as a function of the number of spins $N$ inside the cavity. Blue color (circles) is used for the weak-coupling regime ($\lambda=0.01$), whereas red color (squares) is used for the strong-coupling regime ($\lambda=0.5$). In (c) we plot the entanglement entropy $S^B\coloneqq -\Tr(\rho_B\log_2\rho_B)$ of the battery reduced density matrix, defined as $\rho_B = \Tr_c(\rho)$, where the trace is performed over the cavity degree of freedom. We normalize this quantity by its maximum value  $S^B_{max}=\log_2(N+1)$.}
\label{fig:DickeCapacity}
\end{figure}
In order to be able to make stronger statements about this system, we have numerically solved the dynamics for $\lambda =0.01$ ($\lambda =0.5$), which are representative examples of the weak (strong) coupling regime of the Dicke model. The qualitative difference between these two regimes can be pictorially seen in \figref{fig:dickie}, where we see that the dynamics of energy levels are very local (and thus similar to the parallel charging case) in the weak-coupling regime, whereas in the strong coupling regime they exhibit highly non-local properties, with levels well separated in energy get entangled during evolution. 

Let us now turn to the quantitative discussion of these two regimes. A first result [see Fig.~\ref{fig:Dicke2}(c)] is that, when including the normalization $\frac{1}{\sqrt{N}}$ in the coupling, the super-extensive behavior of power presented in Ref.~\cite{PhysRevLett.120.117702} disappears, and linear scaling in $N$ (as in the parallel case) is approximately recovered, both in the weak and the strong-coupling regime. In contrast, when analyzed in terms of the quantities appearing in our bound, the strong-coupling regime shows relevant differences with respect to a parallel charging scenario. 

In \figref{fig:Dicke2}b, we see that in both the weak and strong-coupling regime, the time-averaged Fisher information $I_E$ scales approximately linearly with $N$, in agreement with the linear bound set by $\Delta H_{DK}^2$. The differences appear when looking at the time averaged $\Delta H_B^2$ [see \figref{fig:Dicke2}(a)]. While in the weak-coupling regime it scales linearly with ${N}$, in the strong-coupling regime this quantity is enhanced and close to the $N^2$ scaling. This superlinear scaling is clearly associated to the large value of the coupling $\lambda$, that allows for cavity-mediated interactions between the spins, that explore highly correlated subspaces. In fact, as Eq.~\eqref{eq:var_ent} is also valid for mixed states, this super-linear scaling means that the spins (i.e., the battery) are in a mixture of $k$-qubit entangled pure states, with $k\gg 1$. 

Nevertheless, the enhancement of $\Delta H_B^2$ in the strong-coupling regime is not reflected in the scaling of power, as one can see in \figref{fig:Dicke2}(d) that the bound (\ref{eq:main_thm}) is far from being saturated in this regime, i.e., $ \langle P\rangle_{\Delta t}\ll\sqrt{\langle\Delta H_B^2\rangle_{\Delta t} \langle I_E\rangle_{\Delta t}} $, leading to power scaling only linearly with $N$. Remarkably, our bound is tight at order $N$ in the weak-coupling regime.

Let us finally discuss the Dicke model in terms of storage. We observe in \figref{fig:DickeCapacity} that, while in the weak coupling regime the capacity properties are similar than those of a parallel charging, in the strong coupling regime there is a worsening in the quality of the stored energy. This is because the super-extensive energy variance generated during the charging process does not disappear completely in the final state [see \figref{fig:DickeCapacity}(b)], as they decay much slower than $1/\sqrt{N}$. In \figref{fig:DickeCapacity}(c) we also show the significant presence of entanglement between the final state of the battery and the source (i.e., the cavity) in this strong-coupling regime.

\section{Conclusions \label{sec:Discussion}}
For a quantum battery, in which the cells are quantum in nature and the process of charging and discharging could introduce quantum correlation among them, a natural question is how to harness the quantum advantages such that it outperforms a classical battery. The aim of this work is to find physically meaningful quantities, in relation to the quantum nature of cells and processes, and bounds to characterize a quantum battery. The important properties of a battery are: \emph{capacity}, i.e., the amount of energy it can store and deliver; \emph{power}, that signifies how fast a battery can be charged or discharged; and \emph{variance} in the stored energy, which determines the quality of the stored energy and to which extent it can be deterministically accessed. 
For a battery composed of many noninteracting identical quantum cells, the capacity is additive and it is independent of the correlations present in the battery state. While in the case of power, which may not be additive, we expect to see roles of inter-cell correlations leading to certain quantum advantages.

We have derived the capacity, that is, a fundamental bound for the storage of a quantum battery, with the help of the energy-entropy diagram. For a battery with a finite number of quantum cells and a general unitary charging and discharging process, the capacity often does not saturate. However, in the thermodynamic limit (i.e., with a considerably large number of quantum cells) the capacity is saturated. 

While studying the power of a quantum battery, we have considered the evolution of quantum states when it is projected in the eigenspace of the battery Hamiltonian. Such an approach led us to derive a lower bound on power in terms of two battery dependent quantities. One of these quantities is the Fisher information calculated after the battery state is projected in the eigenspace of the battery Hamiltonian, and the other one is the energy variance of the battery. While both these quantities are influenced by the appearance of entanglement like correlations during charging (discharging) processes, the former signifies how fast the process takes place in the energy eigenspace and latter encodes how smart (in terms of path length) the trajectory of evolution is. This approach enables us to characterize quantum advantages arising from two different sources. Moreover, for the case of qubit-based quantum batteries, we have presented an inequality that shows the maximum charging power that a battery can exhibit when it has at most $k$-qubit entanglement, thus establishing a novel interrelation between entanglement and power. The fact that the $k$-producibility is related with the violation of Bell inequalities \cite{PhysRevLett.123.100507, PhysRevA.100.032307} could also lead to further studies, in which the relation of power with nonlocality \cite{PhysicsPhysiqueFizika.1.195} is explored.

In the light of the newly introduced bounds, we have considered several spin models of quantum batteries. A first set of paradigmatic models has served us to illustrate how the saturation of the presented bounds can be achieved, and the typical scalings with the number of cells (that is, the spins) of the relevant quantities. We have also studied a few physically realizable models. While exhibiting a reasonably good behavior in terms of capacity, they are not able to show a super-linear scaling in power. In the case of integrable spin models, despite their interacting nature, the hidden local structure in momentum space makes them very similar to a battery in which the cells are charged independently. For the LMG model, we have shown that quantum entanglement enhances the charging power, but the reduced speed of evolution $I_E$ leads to the same overall scaling as in the case where the cells are charged independently. Finally, we have also studied the Dicke model in two regimes. In the weak-coupling regime, the model is similar to the \textit{parallel} case, while in the strong-coupling regime, there is a large generation of entanglement between the spins. However, our bound is poorly staurated in this latter regime, and entanglement does not contribute in enhancing the power. That is why it exhibits the same scaling as in the \textit{parallel} case.

As a final remark, we believe that our approach and results towards characterizing quantum batteries, in terms of bounds on storage and power, and correct identification of the origins of quantum advantages, will find important applications in its theoretical and technological aspects.

\section*{Acknowledgments}
We acknowledge the Spanish Ministry MINECO (National Plan 15 Grant: FISICATEAMO No. FIS2016-79508-P, SEVERO OCHOA No. SEV-2015-0522, FPI), European Social Fund, Fundaci\'o Cellex, Generalitat de Catalunya (AGAUR Grant No. 2017 SGR 1341 and CERCA/Program), EU FEDER, ERC AdG OSYRIS, EU FETPRO QUIC, and the National Science Centre, Poland-Symfonia Grant No. 2016/20/W/ST4/00314. M.N.B. gratefully acknowledges financial supports from Max-Planck Institute fellowship and from SERB-DST, Government of India, and A.R. thanks supports from CELLEX-ICFO-MPQ fellowship. We also thank M. Polini for fruitful discussions and for drawing our attention to Ref. \cite{PhysRevLett.103.083601}, and an anonymous referee for motivating the finding of Corollary \ref{cor:PowerEnt}.

\appendix

\section{Solution to the dynamics of the integrable spin models \label{App:A}}

Here we recast the Hamiltonian 

\be \label{eq:hamiltonianSpin}
 \begin{split}
 H_{JW} &= H_B +\frac{1}{2}\sum_{\substack{j=0 \\ m=1}}^{N-1}\left[(\lambda_m+\gamma_m)\sigma^{j}_x(\otimes_{l=j+1}^{j+m-1}\sigma^l_z)\sigma^{j+m}_x+\right.\\
 &+\left.(\lambda_m-\gamma_m)\sigma^{j}_y(\otimes_{l=j+1}^{j+m-1}\sigma^l_z)\sigma^{j+m}_y\right],
 \end{split}
 \ee
To diagonalize it we start by mapping it, trough the JW transformation, to a fermionic chain with the battery and charging Hamiltonians in quadratic forms, i.e., 
 \be 
 \begin{split}
 H_B &=\sum_{j} f\dagg_{j}f_{j},\\
 H_{JW}&= H_B+\sum_{\substack{j=0 \\ m=1}}^{N-1}\,\left[\lambda_m(f_jf_{j+m}\dagg-f_j\dagg f_{j+m})+ \right. \\
 +& \left. \gamma_m(f_jf_{j+m}-f_j\dagg f_{j+m}\dagg)\right],
 \end{split}
 \ee
where $f_j\ (f^\dagger_j)$ are the annihilation (creation) fermionic operation at the $j$-th site of the chain, and we have dropped an irrelevant constant from $H_B$. Notice that in the fermionic picture the battery Hamiltonian is the particle number operator, and therefore a charging process occurs trough the creation of particles driven by $H_{JW}$, from the initial vacuum state. It is useful to perform a Fourier transformation of the fermionic operators to bring the battery and charging Hamiltonian into the following form:
\be \label{eq:hbatteryJW}
\begin{split}
H_B=& \sum_{k\in \# \text{BZ}} (f\dagg_{k}f_{k}+f\dagg_{-k}f_{-k}),\\
H_{JW}=&\sum_{k\in \# \text{BZ}} \Psi_k\dagg M_k \Psi_k,
\end{split}
\ee
with the definitions 
\be 
\begin{split}
\Psi_k=(f_k,\ f\dagg_{-k})^T,\\
f_k=\frac{1}{\sqrt{N}}\sum_j e^{-\im kj}f_j,
\end{split}
\ee

Notice that $\#BZ$ stands for the reduced Brillouin zone, that is, the subset of positive $k$'s from the set $k= -\pi+\frac{2\pi}{n}m$, with $m\in [0,N-1]$. $M_k$ is a $2 \times 2$ matrix, which form depends on the parameters of the model and reads

\be 
\begin{split}
M_k&= \omega_k\begin{pmatrix}
	  \cos{\theta_k} & \sin{\theta_k}e^{-i\pi/2}  \\
	\sin{\theta_k}e^{i\pi/2}  &-\cos{\theta_k}
	\end{pmatrix}
\end{split}
\ee

where $\theta_k$ and $\omega_k$ are defined as
\be 
\begin{split}
\label{eq:paramJW}
& \omega_k  = 2 \sqrt{\left[\frac{1}{2}-\sum_m \lambda_m \cos{(km)}\right]^2+\left[\sum_m \gamma_m \sin{(km)}\right]^2},\\
& \sin{\theta_k} = \frac{2\sum_m \gamma_m \sin{(km)}}{\omega_k}.
\end{split}
\ee 

The dynamics can now be solved by performing a simple Bogoliubov transformation. To do so, we define the Bogoliubov fermionic modes as $\Phi_k = U_k\Psi_k$, where $U_k$ is a $2\times 2$ unitary matrix. The matrix $U_k$ is chosen such that the Hamiltonian $H_{JW}$ in \eqnsref{eq:hbatteryJW} becomes diagonal, when rewritten in terms of the newly defined $\Phi_k$ operators (i.e., the matrix $U_kM_kU_k^\dagger$ is a diagonal matrix for all $k$). This matrix reads

\be \label{def:UkXX}
 \begin{split}
 U_k & = \begin{pmatrix}
 e^{-\im \pi/4}\cos(\theta_k/2)  & e^{-\im \pi/4}\sin(\theta_k/2) \\
	  -e^{+\im \pi/4}\sin(\theta_k/2)\ & e^{+\im \pi/4}\cos(\theta_k/2)  
	\end{pmatrix}, 
 \end{split}
 \ee
and the charging Hamiltonian takes the form 

\be \label{def:UkXX}
 \begin{split}
 H_{JW} & = \sum_{k\in \# \text{BZ}}\omega_k\Phi_k^\dagger\begin{pmatrix}
 1  & 0 \\
	  0 & -1  
	\end{pmatrix}\Phi_k. 
 \end{split}
 \ee
 From the last expression, it is easy to see that, in the Heisenberg picture, the dynamics generated by $H_{JW}$ will lead to 
 
 \be 
 \begin{split}
 \Phi_k(t)= &  \begin{pmatrix}
 e^{-\im \omega_k}  & 0 \\
	  0 & e^{+\im \omega_k}  
	\end{pmatrix}\Phi_k(0). 
 \end{split}
 \ee
 
 Finally, one can then transform back to the Fourier transformed fermionic modes, 
 
 \be 
 \Psi_k(t) = U_k^\dagger \Phi_k(t),
 \ee
 
 and compute all the relevant dynamical quantities of the problem.

\be 
\begin{split}
E(t)&=\sum_{k\in \# \text{BZ}} \varepsilon_k(t), \hspace{4mm} P(t)= \sum_{k\in \# \text{BZ}} \dot{\varepsilon}_k(t),\\
\Delta H_B(t)^2&=\sum_{k\in \# \text{BZ}} {\varepsilon}_k(t)(2-{\varepsilon}_k(t)),\hspace{4mm}\Delta H_{JW}^2=\sum_{k\in \# \text{BZ}} \sin^2(\theta_k)\omega^2_k, 
\label{eq:jordan}
\end{split}
\ee

where  $\varepsilon_k(t) =2\sin^2{\theta_{k}}\sin^2{\omega_k t}$.

\section{Fisher information of integrable spin models \label{App:B}}
We want to compute the Fisher information for the Hamiltonians of integrable spin models, considered in Eq.~\eqref{eq:hamiltonianSpin}. To do so, we first notice that in the fermionic picture the coefficient $p_l(t)$ that appears in the definition of the Fisher information is the probability of having $l$ particles at a time $t$. We also notice that $p_l(t)$ will only have a non-zero value for even $l$, as particles need to be created in pairs $(k,-k)$ to conserve momentum. Finally, the local structure of the problem in $k$-space allows us to compute separately for every $(k,-k)$ subspace the probability of being in the vacuum state or in the pair state at a given instant of time. It is easy to see that
\be 
p_{k,-k}(t)=\frac{\varepsilon_k(t)}{2}.
\ee
To compute the Fisher information we use these local probabilities to determine the probability that the hole state contains $l$ particles:
\be 
p_l(t) = \sum_{\{\sigma\}_l} p(\sigma_1)\dots p(\sigma_N)
\ee
where the sum runs over the configurations $\{\sigma\}_l$, which are binary strings, each bit representing the state of a $(k,-k)$ subspace, with 0's (1's) representing the vacuum (pair) state. These configurations are constrained to contain $l/2$ pairs (1's) and 
\be 
p(\sigma_k) =\begin{cases} 
      \varepsilon_k/2 & \text{if }\sigma_k=1, \\
      1-\varepsilon_k/2 & \text{if }\sigma_k=0. 
     \end{cases}
\ee
And the time derivative of the above probability reads
\be 
\dot{p}_l(t) = \sum_{\{\sigma\}_l} \sum_{\sigma_k} \frac{\dot{\varepsilon}_{k}}{2}(-1)^{1+\sigma_k}\prod_{\sigma_{j\neq k}}p(\sigma_j).
\ee

\bibliographystyle{apsrev4-1}

\end{document}